\documentclass[12pt,a4paper]{article}

\usepackage[utf8]{inputenc}
\usepackage{lmodern}
\usepackage{amssymb}
\usepackage{amsthm}
\usepackage{amsmath,amsfonts}
\usepackage{mathabx}
\usepackage[top=3cm,bottom=3cm,left=2.5cm,right=2.5cm]{geometry}
\usepackage{harvard}
\usepackage{enumitem}
\usepackage{hyperref}
\usepackage{bbm}
\hypersetup{
    colorlinks=false,
    pdfborder={0 0 0},
}

\usepackage{algorithmic}
\usepackage{float}
\algsetup{linenosize=\scriptsize,linenodelimiter=.,indent=1em}
\floatstyle{ruled}
\newfloat{algorithm}{h}{}
\floatname{algorithm}{Algorithm}

\allowdisplaybreaks[3]

\usepackage[center, small]{caption}
\usepackage{subfig}

\usepackage{tikz}
\usetikzlibrary{patterns}
\usetikzlibrary{arrows,backgrounds,calc,decorations.pathreplacing}
\colorlet{myblue}{blue!80!green}
\colorlet{mybluelight}{myblue!50}
\tikzset{
  > = latex',
  axis/.style    = {very thick},
  aborder/.style = {draw},
  acomp/.style   = {fill=black, fill opacity=0.1},
  rect/.style    = {very thick},
  form/.style    = {font=\scriptsize},
  sm/.style      = {font=\small},
  vsm/.style     = {font=\scriptsize}
}

\newcommand{\E}{\mathbb{E}}
\newcommand{\Ei}{\mathbb{E}_{t_{-i}}}
\newcommand{\Eis}{\mathbb{E}_{t_{-i},s}}
\newcommand{\ti}{t_{-i}}
\newcommand{\1}{\mathbf{1}}

\DeclareMathOperator*{\essinf}{ess\, inf}
\DeclareMathOperator*{\esssup}{ess\, sup}

\DeclareMathOperator{\Prob}{Prob}

\newtheorem{theorem}{Theorem}
\newtheorem*{theorem*}{Theorem}
\newtheorem{lemma}{Lemma}
\newtheorem{definition}{Definition}

\newtheorem{remark}{Remark}
\newtheorem{proposition}{Proposition}

\newtheorem{example}{Example}
\newtheorem{claim}{Claim}

\newcommand{\iec}{i.\,e.}
\newcommand{\eg}{e.\,g.}

\newcommand{\I}{\mathcal{I}}

\begin{document}

\title{Costly Verification in Collective Decisions\footnote{ This paper was previously circulated under the title ``Optimal Social Choice with Costly Verification''. We are grateful to two anonymous referees, Simon Board (coeditor), Daniel Kr\"{a}hmer and Benny Moldovanu for detailed comments and discussions. We would also like to thank Hector Chade, Eddie Dekel,  Francesc Dilm\'{e}, Navin Kartik, Jens Gudmundsson, Bart Lipman and seminar participants at Bonn University and at various conferences and universities for  insightful comments. Albin Erlanson gratefully acknowledges financial support from the European Research Council and the Jan Wallander and Tom Hedelius Foundation. Erlanson: Department of Economics at University of Essex,  \texttt{albin.erlanson@essex.ac.uk}; Kleiner: Department of Economics, W.P. Carey School of Business at Arizona State University, \texttt{andreas.kleiner@asu.edu}.}}
\author{Albin Erlanson \and Andreas Kleiner}
\date{\today}
\maketitle

\begin{abstract}

We study how a principal should optimally choose between implementing a new policy and maintaining the status quo when information relevant for the decision is privately held by agents. Agents are  strategic in revealing their information; the principal cannot use monetary transfers to elicit this information, but can verify an agent's claim at a cost. 
We characterize the mechanism that maximizes the expected utility of the principal. 
This mechanism can be implemented as a cardinal voting rule, in which
agents can either cast a baseline vote, indicating only whether they are in favor of the new policy, or they make specific claims about their type. The principal gives more weight to specific claims and verifies a claim whenever it is decisive.

\\

\end{abstract}

\newpage

\section{Introduction}\label{sec_intro}

The usual mechanism design paradigm assumes that agents have private information and the only way to learn this information is by giving agents incentives to reveal it truthfully.
This is a suitable model for many situations, most importantly when agents have private information about their preferences. But  there are a number of environments where agents' private information is based on hard facts. This could enable an outside party to learn the private information of the agents, at a potentially significant cost.

For example, consider a CEO in a company who faces an investment decision. Board members have relevant information but could have misaligned incentives because the investment has different effects on different divisions. The CEO can take the information provided by a board member at face value or hire consultants to check, various claims made by a board member.  Another example are large mergers in the EU, which must be approved by the European Commission. If a proposed merger has a potentially large impact and its evaluation is not clear, a detailed investigation is initiated. The Commission collects information from the merging companies, third parties and competitors. According to the Commission, this investigation ``typically involves more extensive information gathering, including companies' internal documents, extensive economic data, more detailed questionnaires to market participants, and/or site visits''.  The analyses carried out by the Commission on potential efficiency gains requires that ``claimed efficiencies must be verifiable'' \cite{EuropeanUnion}.
Lastly consider an example, taken from Sweden,  on the decision of whether a newly approved pharmaceutical drug should be subsidized. A producer of a drug can apply for a subsidy by providing arguments for the clinical and cost-effectiveness of the drug. Other stakeholders  are also given an opportunity to participate in the deliberations by contributing 
information relevant to the decision. Importantly, the applicant and other stakeholders should provide documentation supporting their claims \cite{TLV}. 

In order to study such situations, we formulate a model with costly verification in which a principal  decides between introducing a new policy and maintaining the  status quo. The principal's optimal choice depends on agents' private information, summarized by each agent $i$'s type $t_i\in \mathbb{R}$. Agents can be in favor of or against the new policy, and they are strategic in revealing their information since it influences the decision made by the principal. We exclude monetary transfers, but before taking the decision the principal can verify any agent and learn his information at a cost $c_i$. We determine the mechanisms that maximizes the expected payoff of the principal; it optimally  solves  the trade-off between the benefits from using detailed information as input to the decision rule and the implied costs from verifying agents' claims to make the mechanism  incentive compatible.

In the optimal mechanism, agents can vote in favor or against the new policy; moreover, they have the option to report their exact type. If agent $i$ reports his type, the principal adjusts the reported type by the verification cost $c_i$ to obtain agent $i$'s \emph{net type}, which is $t_i-c_i$ if $i$ votes in favor and $t_i+c_i$ if he votes against (see Figure \ref{fig:intro} for an illustration). If an agent does not report his type the principal assumes this agent has a default net type, namely $\omega_i^+$ if he voted in favor of the new policy and $\omega_i^-$ if he voted against. This induces \emph{bunching}, since an agent who is in favor only reports his type if it is high enough and otherwise only casts a vote (and conversely if he is against). The optimal decision rule for the principal is then to implement the new policy whenever the sum of net types is positive. A report is \emph{decisive} whenever it changes the decision compared to this agent not sending a report; in the optimal mechanism each decisive report is verified.

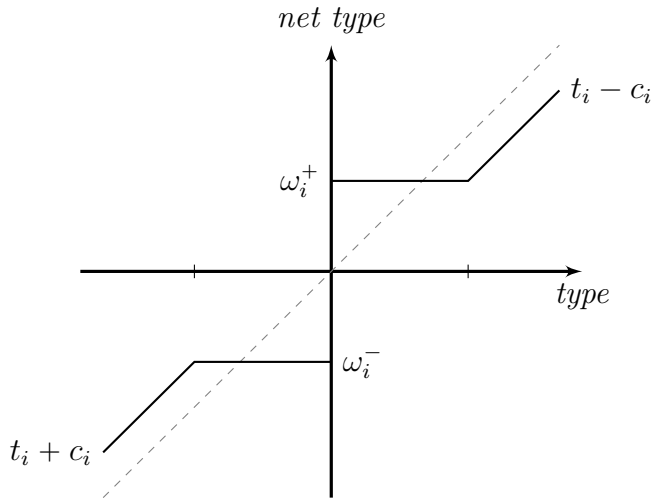
\begin{figure}[htbp!]%
\begin{center}
  \begin{tikzpicture}[decoration={brace},baseline,x=3cm,y=3cm]
      \draw[axis,->] (-1.1,0) -- (1.1,0) node[below] {\textit{type}};
      \draw[axis,->] (0,-1) -- (0,1) node[above] {\textit{net type}};
      \draw[dashed,gray] (-1,-1) -- (1,1) ;
      \draw[thick] (-1,-.8) node[left]{$t_i+c_i$} -- (-.6,-.4) -- (-0,-.4) node[right]{$\omega_i^-$};
      \draw[thick]  (0,.4) node[left]{$\omega_i^+$}  -- (.6,.4) -- (1,.8) node[right]{$t_i-c_i$};
      \draw (0,-.3) ;
      \draw (-.6,-.03) -- (-.6,.03) ;
      \draw (.6,-.03) -- (.6,.03) ;
  \end{tikzpicture}
\end{center}
\caption{Illustration of how types are transformed to net types. The principal implements the new policy whenever the sum of net types is positive.}%
\label{fig:intro}%
\end{figure}

Our analysis provides at least two important insights for the design of mechanisms in applications similar to our model.  We will illustrate them by connecting our analysis to the European Commission's decision on whether to approve a merger. In a merger review, the Commission 
``analyses claimed efficiencies which the companies could achieve when merged together. If the positive effects of such efficiencies for consumers would outweigh the mergers' negative effects, the merger can be cleared'' \cite{EuropeanUnion}. Our analysis suggests, first, that the Commission should not always use claimed efficiencies (which must be verifiable), but might benefit by assuming that a merger has a predetermined estimated efficiency gain, even if they provide no verifiable documentation. Moreover, it might be beneficial to discount efficiency claims that are difficult and expensive to verify.
Second, before starting the process of verifying claimed efficiencies and other reports, the Commission should first determine which reports are decisive and subsequently verify only those. 
While there are other verification rules that could be used,
this is a particularly simple rule that is easy to implement in practice and it provides robust incentives for truth-telling.

To explain the intuition behind the optimal mechanism, we now describe in more detail our main results. We show first that the principal can, without loss of generality, use an incentive compatible direct mechanism, which can be implemented as follows. In the first step, agents communicate their information. For each profile of reports, a mechanism then provides answers to three questions: First, which reports should be verified (\emph{verification rule})?  Second, what is the decision regarding the new policy (\emph{decision rule})? Finally, what is the penalty when someone is revealed to be lying? Because we can focus on incentive compatible mechanisms, penalties will be imposed only off the equilibrium path. The principal can therefore always choose the severest possible penalty, as this weakens incentive constraints but does not affect the decision made on the equilibrium path. 
In general, the principal can implement any decision rule by always verifying all agents. However, the principal has to make a trade-off between using detailed information for ``good'' decisions and incurring the costs of verification. 

Key to solving the principal's problem is that incentive constraints are tractable. Each agent wants to send the report that maximizes the probability that his preferred decision is implemented. We show that if there is a profitable deviation for some type, any type that has a lower equilibrium probability of getting his preferred outcome also finds this deviation profitable. This suggests that incentive constraints are hardest to satisfy for the types that have the lowest equilibrium probability of getting their preferred decision; we call these types the \emph{worst-off types}.\footnote{Since we allow for general utility functions, these are not necessarily the types with the lowest expected utility.} It follows that a mechanism is incentive compatible if and only if it is incentive compatible for the worst-off types. 

We can now explain how and why the optimal mechanism differs from the first best outcome.
First, in the optimal mechanism the principal incurs costs of verification. Verifications are clearly necessary if information is private and, since the incentive constraints for worst-off types are exactly binding, the optimal mechanism uses costly verifications as rarely as possible. 
Second, the decision is distorted compared to the first-best because there is bunching at the bottom. This is optimal for the principal because, as observed above, incentive constraints are hardest to satisfy for worst-off types. Suppose instead there was no bunching at the bottom and a single type had the lowest probability of getting the preferred decision. Then any higher report has to be verified sometimes to make the worst-off type indifferent between reporting truthfully and deviating. Now if we increase the probability that the worst-off type gets his preferred outcome this will only change the decision for this type, which has essentially no effect on the principal's expected utility from the decision. But this makes it less attractive for the worst-off type to claim to be of a different type and the principal can, therefore, verify all other types with a strictly lower probability. Thus, this change allows the principal to save on verification costs for almost all reports but it only  changes the decision for one type. This  implies that the cost-saving effect dominates. We conclude that the original mechanism, with a single worst-off type, could not have been optimal and that the optimal mechanism must feature bunching at the bottom. 
Finally, the principal's first-best decision would be to implement the new policy whenever the sum of types is positive, but in the optimal mechanism the principal uses net types instead to determine the decision, which introduces a further distortion. Whenever an agent's report $t_i$ is verified, the principal pays the verification cost $c_i$. If the principal implements the new policy because agent $i$ reported a high type, $i$'s effect on the principal's payoff is only his net value $t_i-c_i$ and not his actual type $t_i$ because the principal has to pay the verification cost $c_i$. It is, therefore, optimal for the principal to distort the decision rule by using net types instead of true types.

The remainder of the paper is organized as follows. After reviewing relevant literature, we present in Section~\ref{sec_model} our main model and describe the principal's objective. In Section~\ref{sec_cutoff}, we discuss the optimal mechanism, and in Section 4 we prove an equivalence between Bayesian and ex-post incentive compatible mechanisms. We consider various extensions in Section~\ref{sec_rob_ext}, including an analysis of the optimal mechanism with imperfect verification. All proofs not found in the main body of the paper are relegated to the Appendix.

\subsection*{Related Literature}\label{sec_rel_literature}
 
There is a substantial  literature on collective choice problems with two alternatives when monetary transfers are not possible. A particular strand of this literature, dating back to the seminal work of \citeasnoun{rae69}, assumes that agents have cardinal utilities and compares decision rules with respect to ex-ante expected utilities. Because money cannot be used to elicit cardinal preferences, Pareto-optimal decision rules are simple and can be implemented as voting rules, where agents indicate only whether they are in favor of or against the policy \cite{schmitzTroger12,azreliKim14}. Introducing a technology to learn the agents' information allows  a much richer class of decision rules to be implemented. Our main interest lies in understanding how this additional possibility allows for other implementable mechanisms and  changes the optimal decision rule.

\citeasnoun{townsend79} introduces costly verification in a principal-agent model with a risk-averse  agent. Our model differs from his, and the literature building on it \citeaffixed{galeHellwig85,borderSobel87}{see \eg}, since monetary transfers are not feasible in our model. Allowing for monetary transfers yields different incentive constraints and economic trade-offs than in a model without money.

Recently, there has been growing interest in models with state verification that do not allow for transfers. \citeasnoun[henceforth BDL]{ben-porath14} consider a principal that wishes to allocate an indivisible good among a group of agents, and each agent's type can be learned at a given cost. The principal's trade-off is between allocating the object efficiently and incurring the cost of verification. BDL characterize the mechanism that maximizes the expected utility of the principal: it is a favored-agent mechanism, where a pre-determined favored agent receives the object unless another agent claims a value above a threshold, in which case the agent with the highest (net) type gets the object. We study a similar model of costly verification and without transfers, but we are interested in optimal mechanisms in collective choice problems. In these problems more complex voting mechanisms are feasible, even in the absence of verification possibilities. 
More recently, \citeasnoun{mylovanov17} study the allocation of an indivisible good when the principal always learns the private information of the agents but only after having made the allocation decision and having only limited penalties at his disposal. \citeasnoun{halac17} introduce costly verification in a delegation setting and describe the  conditions under which interval delegation with an ``escape clause'' is optimal.

\citeasnoun{glazerRubinstein04} and \citeasnoun{glazerRubinstein06}  consider a situation in which an agent has private information about several characteristics and tries to persuade a principal to take a given action, and  the principal can only check one of the agent's characteristics. Recently, \citeasnoun{ben-porath17} study a class of mechanism design problems with evidence. They show that the optimal mechanism does not use randomization, commitment is not an issue, and robust incentive compatibility does not entail any cost. Additionally, they show that costly verification models  can be embedded as evidence games as an alternative way of finding optimal mechanisms, but the results on commitment and robustness does not apply to costly verification models.\footnote{For additional  papers on mechanism design with evidence, see also \citeasnoun{greenLaffont86}, \citeasnoun{bullWatson07}, \citeasnoun{deneckereSeverinov08}, \citeasnoun{benporathLipman12}.}

\section{Model and Preliminaries}\label{sec_model}

There is a principal and a set of agents $\I=\{1,2,\dots,I\}$. The principal decides between implementing a new policy and maintaining the status quo.\footnote{We discuss in Section \ref{sec_rob_ext} how our analysis changes if the principal can decide between more than two actions.} Each agent holds private information, summarized by his type $t_i\in \mathbb{R}$. The payoff to the principal is $\sum_i t_i$ if the new policy is implemented, and it is normalized to zero if the status quo remains. Monetary transfers are not possible. The private information held by the agents is verifiable. The principal can check agent $i$ at a cost of $c_i$, in which case he  learns the true type of agent $i$. Being verified imposes no costs on the agent. Agent $i$ with type $t_i$ obtains a utility of $u_i(t_i)$ if the policy is implemented and zero otherwise. For example, if $u_i(t_i)=t_i$ for each agent, the principal maximizes utilitarian welfare; in general, the principal could have divergent preferences, for example, because he only cares about how the new policy affects himself.\footnote{Another interpretation of the objective function, suggested by a referee, is that the principal is interested in the mean of an unknown parameter.} Types are drawn independently from the type space $T_i \subset \mathbb{R}$ according to the distribution function $F_i$ with finite moments and density $f_i$. Let $t\equiv (t_i)_{i\in \I}$ and $T\equiv \prod_i T_i$. 
 
The principal can design a mechanism and agents play a Bayesian Nash equilibrium in the game induced by the mechanism. A mechanism could potentially be an indirect and complicated dynamic mechanism that includes multiple rounds of communication and checking. However, we show in Appendix \ref{sec_rev_principle} that it is without loss of generality to focus on direct mechanisms with truth-telling as a Bayesian Nash equilibrium. To allow for stochastic mechanisms we introduce a correlation device as a tool to correlate the decision rule with the verification rules. Assume that $s$ is a random variable that is drawn independently of the types from a uniform distribution on $[0,1]$, and only observed by the principal. 
A direct \textit{mechanism} $(d,a,\ell)$ consists of a \textit{decision rule} $d:T\times[0,1]\rightarrow\{0,1\}$, a profile of \textit{verification rules} $a\equiv(a_i)_{i\in\I}$, where $a_i:T\times[0,1]\rightarrow\{0,1\}$, and a profile of \textit{penalty rules} $\ell\equiv(\ell_i)_{i\in \I}$, where $\ell_i:T \times T_i \times [0,1]\rightarrow\{0,1\}$. 
In a direct mechanism $(d,a,\ell)$, each agent sends a message $m_i \in T_i$ to the principal. Given these messages the principal verifies agent $i$ if $a_i(m,s)=1$. If no one is found to have lied, the principal implements the new policy if $d(m,s)=1$.\footnote{With slight abuse of notation, we will drop the realization of the randomization device as an argument whenever the correlation is irrelevant. In these cases, $\E_{s}[d(m,s)]$ is simply denoted as $d(m)$ and $\E_{s}[a_i(m,s)]$ is denoted as $a_i(m)$.}  If the verification reveals that agent $i$ has lied, the new policy is implemented if and only if $\ell_{i}(m,t_{i},s)=1$, where $t_{i}$ is agent $i$'s true type. If more than one agent lied it is arbitrary what decision to take. For each agent $i$, let $T_i^+ := \{t_i \in T_i | u_i(t_i) >0\}$ denote the set of types that are in favor of the new policy, and let $T_i^- := \{t_i \in T_i | u_i(t_i) <0\}$ denote the set of types that are against the policy. We assume that $t_i^- < t_i^+$ for all $t_i^- \in T_i^-$ and $t_i^+ \in T_i^+$. This assumption ensures a weak alignment between the agents' and the principal's preferences: if an agent is in favor of the new policy this increases the principal's expected utility from implementing the policy. This implies that no agent has an incentive to misrepresent his ordinal type, for example by claiming that he is in favor of the new policy while he actually is against the new policy. To simplify notation, we also assume that no agent is indifferent, so $T_i = T_i^+ \cup T_i^-$.

\begin{example}\rm\label{ex_utility}
To illustrate a situation in which the general utility function $u_i(t_i)$ is useful, consider a principal deciding whether to provide a public good. Agents are privately informed about their value for the public good, which is always positive. The principal bears the cost $k$ of providing the public good and maximizes the sum of the agents' values minus the potential cost of providing the public good.

This example can be mapped into our model by defining the type of an agent that values the public good at $v_i$ to be $t_i=v_i-k/I$, and by setting $u_i(t_i)>0$ for all $t_i$. Clearly, all agents are in favor, even if their types are negative, and the principal's payoff of providing the public good is $\sum t_i = \sum v_i - k$.

\end{example}

\medskip

Truth-telling is a Bayesian Nash equilibrium for the mechanism $(d,a,\ell)$ if and only if the mechanism $(d,a,\ell)$ is Bayesian incentive compatible, which is formally defined as follows.

\begin{definition}\label{def_bic}\rm
A mechanism $(d,a,\ell)$ is \emph{Bayesian incentive compatible (BIC)} if, for all $i\in \I$ and all $t_i,t_i'\in T_i$,
\begin{align*}
     u_i(t_i) \cdot \E_{t_{-i},s}[d(t_i,t_{-i},s)] \ge u_i(t_i) \cdot  \E_{t_{-i},s}[d(t'_i,t_{-i},s) [1-a_i(t'_i,t_{-i},s)] + a_i(t'_i,t_{-i},s) \ell_i(t'_i,t,s)].
\end{align*}
\end{definition}
The left-hand side of the equation in Definition~\ref{def_bic} is the interim expected utility if agent $i$ truthfully reports his type $t_i$  and all others also report truthfully. The right-hand side is the interim expected utility if agent $i$ instead lies and reports to be of type $t'_i$.

The aim of the principal is to find an incentive compatible mechanism that maximizes his expected utility. The expected utility of the principal for a given mechanism $(d,a,\ell)$ is
\[
\E_{t}\Big[ \sum_{i} (d(t)t_i - a_i(t)c_i) \Big], 
\]
where expectations are taken over the prior distribution of types. 

Because the principal uses an incentive compatible mechanism, lies will occur only off the equilibrium path and will therefore not directly enter the objective function. The principal can therefore always choose the severest possible penalty for a lying agent. This will not affect the outcome on the equilibrium path, but it weakens the incentive constraints. For example, if an agent is found to have lied and his true type supports the new policy, the penalty will be to maintain the status quo. Henceforth, without loss of optimality,  we assume that the principal uses this penalty scheme and, we will drop the reference to a profile of penalty functions when we describe a mechanism. 

At this point, we have all the prerequisites and definitions required to formally state the aim of the principal:   
\begin{align*} \label{P}
	\underset{d,a}\max \ &\E_{t}  \Big[\sum_i(d(t)t_i - a_i(t)c_i)\Big]  \tag{P}\\
	& \text{s.t. $(d,a)$ is Bayesian incentive compatible.}
\end{align*}

The following lemma provides a characterization of Bayesian incentive compatible mechanisms.

\begin{lemma}\label{lemma_bic}
A mechanism $(d,a)$ is Bayesian incentive compatible if and only if, for all $i\in \I$ and all $t_i\in T_i$,
\begin{align*}
     \inf_{t_i'\in T_i^+} \E_{t_{-i},s}[d(t_i',t_{-i},s)] &\ge \E_{t_{-i},s}[d(t_i,t_{-i},s) [1-a_i(t_i,t_{-i},s)]]  \hspace{.5cm}\text{ and } \\
      \sup_{t_i'\in T_i^-} \E_{t_{-i},s}[d(t_i',t_{-i},s)] &\le \E_{t_{-i},s}[d(t_i,t_{-i},s) [1-a_i(t_i,t_{-i},s)] + a_i(t_i,t_{-i},s)].
\end{align*}
\end{lemma}

\noindent
We call a type a \emph{worst-off type} if the infimum (respectively the supremum) in Lemma \ref{lemma_bic} is attained for this type. The intuition for Lemma \ref{lemma_bic} is as follows: first, because of the binary nature of the principal's decision an agent maximizes his utility by sending a report that maximizes the probability of getting the preferred decision. Now if type $t_i$ can increase this probability by deviating to a report $t_i'$, any other type can use the same deviation $t_i'$ to get the same probability (since types are distributed independently). By construction worst-off types have the lowest probability of getting their preferred decision when being truthful.  Thus,  whenever some type has a profitable deviation so does the worst-off types.

\begin{proof} [Proof of Lemma \ref{lemma_bic}]
Let $i\in \I$. We will consider two cases, one when agent $i$ is in favor of the policy ($t_i'\in T_i^+$), and the other case is when agent $i$ is against the policy ($t_i'\in T_i^-$). 

Since $u_i(t_i)>0$ for $t_i \in T_i^+$ and we can without loss of generality set $\ell_i(t',t_i,s)=0$ for all $t'$ and $t_i \in T_i^+$, we get that agent $i$ with type $t'_i\in T_i^+$ has no incentive to deviate if and only if, 
for all $t_i\in T_i$,
\begin{align}\label{eq_BIC_plus}
     \E_{t_{-i},s}[d(t_i',t_{-i},s)] &\ge \E_{t_{-i},s}[d(t_i,t_{-i},s) [1-a_i(t_i,t_{-i},s)]].
\end{align}
Since~(\ref{eq_BIC_plus}) is required to  hold for all $t'_i\in T_i^+$, it must in particular hold for the infimum over $T_i^+$, which is equivalent to Definition~\ref{def_bic} of BIC. 

Similarly, since $u_i(t_i)<0$ for $t_i \in T_i^-$ and we can without loss of generality set $\ell_i(t',t_i,s)=1$ for all $t'$ and $t_i \in T_i^-$, a type $t_i'\in T_i^-$, has no incentive to deviate if and only if, for all $t_i\in T_i$,
\begin{align}\label{eq_BIC_minus}
      \E_{t_{-i},s}[d(t_i',t_{-i},s)] &\le \E_{t_{-i},s}[d(t_i,t_{-i},s) [1-a_i(t_i,t_{-i},s)] + a_i(t_i,t_{-i},s)].
\end{align}
Since~(\ref{eq_BIC_minus}) is required to  hold for all $t'_i\in T_i^-$, it must in particular hold for the supremum over $T_i^-$, which is equivalent to Definition~\ref{def_bic} of BIC.
\end{proof}

\section{Voting-with-evidence}\label{sec_cutoff}

In this section, we  show that a  voting-with-evidence mechanism is optimal, find optimal weights in a setting with two agents, and discuss comparative statics. 

\subsection{Optimal mechanism}\label{subsec:opt}
To formally define a voting-with-evidence mechanism, we define, given a collection of weights $\{ \omega_i^+ , \omega_i^-\}_{i \in \mathcal{I}}$ satisfying $\omega_i^- \le \omega_i^+$, the \emph{weight function} $w_i:T_i\rightarrow \mathbb{R}$ by
\begin{align*}
  w_i(t_i) = \begin{cases}
	  t_i + c_i &\text{ if } t_i \in T_i^- \text{ and } t_i < \omega_i^{-}- c_i\\    
	\omega_i^{-} &\text{ if } t_i \in T_i^- \text{ and }  t_i \ge \omega_i^{-}-c_i \\        
    \omega_i^{+} \hspace{.3cm} &\text{ if }t_i \in T_i^+ \text{ and }  t_i \le \omega_i^{+}+c_i \\
   t_i - c_i &\text{ if } t_i \in T_i^+ \text{ and } t_i > \omega_i^{+}+ c_i.
  \end{cases}
\end{align*}
Given the weight functions $w_i$, we say that a mechanism is a \emph{voting-with-evidence mechanism} if
\begin{align*}
  d(t) = \begin{cases}
    1 \hspace{.3cm} &\text{ if } \sum{w_i(t_i)} > 0 \\
    0 &\text{ if } \sum{w_i(t_i)} < 0   
  \end{cases}
\end{align*}
and an agent $i$ is verified if and only if he is decisive. An agent $i$ is \emph{decisive} at a profile of reports $t$ if his preferred outcome is implemented and if the decision were to change if his report was replaced by his relevant cutoff ($\omega_i^+ + c_i $ if he is in favor and $\omega_i^- -c_i$ if he prefers status quo).

\begin{figure}[htbp!]%
\begin{center}
  \begin{tikzpicture}[decoration={brace},baseline,x=3cm,y=3cm]
      \draw[axis,->] (-1.1,0) -- (1.1,0) node[below] {$t_i$};
      \draw[axis,->] (0,-1) -- (0,1) node[above] {$w_i$};
      \draw[dashed,gray] (-1,-1) -- (1,1) node[above]{$u_i(t_i)=t_i$};
      \draw[thick] (-0.9,-.6) -- (-.6,-.3) -- (-0,-.3);
      \draw[thick]  (0,.4) node[left] {$\omega_i^{+}$} -- (.6,.4) -- (1,.8);
      \draw (0,-.3) node[right] {$\omega_i^{-}$};
      \draw (-.6,-.03) -- (-.6,.03) node[above] {$\omega_i^- -c_i$};
      \draw (.6,-.03) -- (.6,.03) node[above] {$\omega_i^+ +c_i$};
      \draw[decorate,decoration={amplitude=5pt,mirror}] (0,-.05) -- (1,-.05) node[midway,below] {\small $T_i^+$};
      \draw[decorate,decoration={amplitude=5pt}] (0,-.05) -- (-1,-.05) node[midway,below] {\small $T_i^-$};
  \end{tikzpicture}
\end{center}
\caption{Example illustrating how weights are determined with utility $u_i(t_i)= t_i$.}%
\label{fig:altered}
\end{figure}
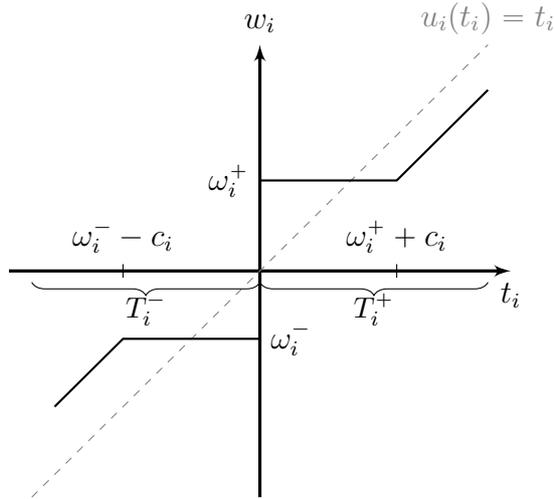

A voting-with-evidence mechanism can be interpreted as a cardinal  voting rule, where agents have the option to make specific claims  to gain additional influence. To see this, consider the following indirect mechanism. Each agent casts a vote either in favor of or against the new policy. In addition, agents can make claims about their information. 
If agent $i$ does not make such a claim, his vote is weighted by the baseline weights $\omega_i^+$ if he votes in favor of  and $-\omega_i^- $ if he votes  against the new policy. If agent $i$ supports the new policy and makes a claim $t_i$, his weight is increased to $t_i-c_i$. Similarly, if he opposes the new policy, his weight is increased to $-t_i+c_i$. The new policy is implemented whenever the sum of weighted votes in favor are larger than the sum of the weighted votes against the new policy. An agent's claim will be checked whenever he is decisive. 
This indirect mechanism indeed implements the same outcome as a voting-with-evidence mechanism. Any agents with weak or no information supporting their desired alternative will prefer to merely cast a vote, whereas agents with sufficiently strong information will make claims to gain additional influence over the outcome of the principal's decision. Note that the cutoffs already determine the default voting rule that is used if all agents cast votes.

A voting-with-evidence mechanism is particularly simple to describe when all agents have type-independent preferences, i.e., for each $i$, $u_i(t_i)> 0 $ or $u_i(t_i)< 0 $ for all $t_i$. For instance, consider the case of deciding on the provision of a public good, where the cost of provision of the public good is  borne by the principal  (this case is spelled out in detail   in  Example~\ref{ex_utility}). Therefore, when agent  $i$ is always in favor of implementing the project,  agent $i$ is assigned a default type of $\omega_i^+ +c_i$, and the principal presumes $i$ has the default type unless $i$ reports differently. The principal reduces the reported (or presumed) type by the verification cost to obtain $i$'s net type, and implements the policy whenever the sum of net types is positive. If an agent changes the outcome because he reports a type different from the default type, he will be verified.

\begin{remark}[Ex-post incentive compatibility of voting-with-evidence mechanisms]\rm\label{rem_expost_cutoff}
We will now show that a voting-with-evidence mechanism is incentive compatible. We will do so by showing that for every type realization truth-telling is a best response. Let $t\in T$ be a profile of types, consider an agent $i$ with type $t_i$, and  assume that agent $i$ is in favor of the new policy, i.e.,  $t_i\in T_i^+$. If $d(t_i,t_{-i})= 1$, then agent $i$ gets his preferred alternative, and there is no beneficial deviation. Suppose instead that $d(t_i, t_{-i})= 0$; then, agent $i$ can only change the decision by reporting some $t_i'> t_i$ and $t_i' > \omega_i^+ +c_i$. However, if $d(t_i',t_{-i})= 1$, then agent $i$ is decisive and will be verified. Agent $i$'s true type $t_i$ will be revealed and the penalty is the retention of the status quo. Thus, agent $i$ cannot gain by deviating to $t_i'$. A symmetric argument holds if agent $i$ is against the new policy, i.e., $t_i\in T_i^-$. These arguments imply that truth-telling is an optimal response to truth-telling for every type realization and therefore independently of the beliefs the agents hold. We conclude that a voting-with-evidence mechanism is \textit{ex-post incentive compatible}.
\end{remark}

\medskip

We are now ready to state our main result.

\begin{theorem}\label{thm_opt}
A voting-with-evidence mechanism maximizes the expected utility of the principal.
\end{theorem}

Appendix~\ref{sec_proofmainresult} contains the proof of Theorem~\ref{thm_opt}. We first prove it for finite type spaces, and then extend the proof  to infinite type spaces through an approximation argument. Before finding optimal weights for a voting-with-evidence mechanism in a two-agent example, we will explain intuitively why these mechanisms are optimal. 

A voting-with-evidence mechanism differs in three respects from the first-best mechanism. We will argue that these inefficiencies have to be present in an optimal mechanism and that any additional inefficiencies will make the principal worse off. First, the principal verifies all decisive agents and incurs the corresponding costs, which he would not need to do if the information were public. Clearly, sometimes verifying agents is necessary to satisfy the incentive constraints for the given decision rule. Moreover, in a voting-with-evidence mechanism the verification rules are chosen such that the incentive constraints are in fact binding: if the principal were to reduce the audit probability for some report, types in the bunching region would have a strict incentive to send this report. Thus, the principal cannot implement the given decision rule with lower verification costs. 

The second inefficiency is introduced by replacing types with net types. Specifically, any report  $t_i \in T_i^+$ and above $\omega_i^+ +c_i$ is replaced by the net type $t_i-c_i$. Similarly, types $t_i \in T_i^-$ and below $\omega_i^- -c_i$ are replaced by the net type  $t_i+ c_i$.
Suppose we replace types of agent $i$ by  net types. Then, for a given profile of types, by replacing agent $i$'s type with his net type, the decision will either remain the same or it will change. First, if the decision remains the same it does not matter whether the type or net type is used. On the other hand, if the decision changes then agent $i$ must be decisive with type $t_i$, but not with the net type. Therefore, the principal has to verify the agent if he uses the type $t_i$ to decide on the policy in order to induce truthful reporting and incurs the cost of verification. Hence, the actual contribution of agent $i$ to the principal's utility is his net type, $t_i - c_i$, and not $t_i$. Thus, the principal is made better off by using $i$'s net type $t_i - c_i$ when determining his decision on the policy, anticipating that he will have to verify the agent whenever he is decisive.

The third inefficiency arises from the fact that all types below the cutoff $\omega_i^+ +c_i$ of an agent in favor of the policy are bunched together and receive the same weight, the baseline weight $\omega_i^+$. Similarly, all types above the cutoff $\omega_i^- -c_i$ and against the policy are bunched together into the baseline weight $\omega_i^-$. Suppose instead that in the optimal mechanism there was a type $t_i'\in T_i^+$ that uniquely had the lowest probability of getting his preferred decision, $\E[d(t_i',t_{-i})]<\E[d(t_i,t_{-i})]$ for all $t_i$. Increasing the probability with which this type gets his most preferred alternative does not affect the principal's expected utility directly (because this type is realized with probability 0). However, our characterization of incentive compatibility implies that changing this probability affects the audit probability for all other types $t_i\in T_i^+$:
\begin{align*}
\Ei[a(t_i,\ti)]\ge \Ei[d(t_i,\ti)] - \Ei[d(t_i',\ti)].
\end{align*}
Therefore, changing the allocation on a Null set will allow the principal to save verification costs with strictly positive probability. This contradicts that the original mechanism could be optimal and implies that any optimal mechanism will have bunching ``at the bottom''.\footnote{More specifically, assume there is an agent $i$ who is always in favor of the new policy, his type space is $T_i=[0,1]$ and suppose $\Ei[d(0,t_{-i})]<\Ei[d(t_i,t_{-i})]$, so $0$ is the only worst-off type. In particular, every report except 0 will sometimes be verified. Consider changing the decision rule so that, for any type $t_i\in[0,\varepsilon]$ and any $t_{-i}$, the probability of implementing the new policy is $\tilde{d}(t_i,t_{-i})=\E[d(z,t_{-i})|z\le \varepsilon]$, and the expected decision is unchanged for all other types of $i$ and all other agents.
It then follows from Lemma \ref{lemma_bic} that for any type above $\varepsilon $ the verification probability can be reduced by $\delta=\Ei[\tilde{d}(0,t_{-i})-d(0,t_{-i})]>0$ and no type of agent $i$ below $\varepsilon$ will ever be verified. For $\varepsilon$ sufficiently small, the saving in verification costs is in the order of $\delta(1- \varepsilon)$ and therefore outweighs the inefficiency induced to the decision rule, which is in the order of $\delta \varepsilon$. Hence, it could not have been optimal to have a unique worst-off type.}

\begin{remark}\rm
We comment briefly on the role of the assumption $t_i^-<t_i^+$ for all $t_i^-\in T_i^-$ and $t_i^+\in T_i^+$. Without this assumption we get a similar result to Theorem \ref{thm_opt} except for the conclusion $\omega_i^+\ge \omega_i^-$. We then have to check whether agents have an incentive to misreport their ordinal preference in this mechanism. As long as $\omega_i^+\ge \omega_i^-$ all incentive constraints are satisfied even if the assumption $t_i^-<t_i^+$ is violated. Only if preferences are strongly misaligned, so that an agent being in favor makes the principal less eager to implement the new policy, we have to augment the mechanism by either (i) verifying agents even if they report in the bunching region or (ii) adjusting the weights so that $\omega_i^+\ge \omega_i^-$ holds.
\end{remark}
%
%
 
\subsection{Optimal weights and comparative statics for two agents}
We will begin with characterizing the optimal weights in an utilitarian setting with two agents and then discuss comparative statics.\footnote{With more than two agents, the weight of agent $i$ not only affects the likelihood that $i$ is decisive, but also has non-trivial effects on the probability that other agents are decisive. It is therefore more difficult to find closed-form solutions for the optimal weights $\omega_i^+,\omega_i^-$.}  
   
\begin{proposition}\label{prop:opt_weights}
Suppose  $I=2$, $T_i^+=\{t_i\in T_i| t_i\ge0\}$, and $T_i^-=\{t_i\in T_i| t_i<0\}$. Let $\omega_i^+$ and $\omega_i^-$ be implicitly defined by 
\begin{align*}
\E[t_i|t_i\ge 0] &= \E[\max\{\omega_i^+,t_i-c_i\}|t_i\ge 0] \text{ and }\\
\E[t_i|t_i< 0] &= \E[\min\{\omega_i^-,t_i+c_i\}|t_i< 0].
\end{align*}
Then voting-with-evidence using weights $\omega_i^+$ and $\omega_i^-$ is optimal.
\end{proposition}

To gain some intuition for the result in Proposition~\ref{prop:opt_weights}, suppose $\omega_1^+> -\omega_2^-$ and consider slightly changing $\omega_1^+$. This only has an effect if $t_2+c_2=-\omega_1^+$, so we condition throughout on this event. If $\omega_1^+$ is slightly increased, then for any $t_1>0$ the project will be implemented and no one will be verified.
On the other hand, if $\omega_1^+$ is slightly decreased there are two cases: if $t_1-c_1+t_2+c_2\ge 0$ the project is implemented and  agent 1 is verified, otherwise the project is not implemented and agent 2 is verified.
We obtain that $\omega_1^+$ satisfies the first-order condition if 
\begin{align*}
\int_0^{\infty} t_1 +t_2 dF_1(t_1) =\int_0^{\infty} (t_1 +t_2-c_1)\mathbf{1}_{t_1-c_1+t_2+c_2\ge 0}(t_1) -c_2 \mathbf{1}_{t_1-c_1+t_2+c_2< 0}(t_2) dF_1(t_1).
\end{align*}
Using $t_2+c_2=-\omega_1^+$, this can be rewritten as
\begin{align*}
\int_0^{\infty} t_1 dF_1(t_1) =\int_0^{\infty} (t_1-c_1)\mathbf{1}_{t_1-c_1\ge \omega_1^+}(t_1) +\omega_1^+ \mathbf{1}_{t_1-c_1<\omega_1^+ }(t_2) dF_1(t_1),
\end{align*}
which yields the first condition in Proposition~\ref{prop:opt_weights}. An analgous argument heuristically explains  the second condition.

Given the characterization of the optimal weights in Proposition~\ref{prop:opt_weights}, we can study how a change in the cost parameter $c_i$ affects the optimal weights. Suppose that the cost of verifying agent $i$ increases. Then the optimal weight $\omega_i^+$ will increase in order for $\E[t_i|t_i\ge 0]$ to equal $\E[\max\{\omega_i^+,t_i-c_i\}|t_i\ge 0]$. Analogously, the increase in $c_i$ implies that $\omega_i^-$ decreases. We conclude that, as the cost of verifying an agent increases, the bunching region increases and the agent will be verified less often. Another possible comparative static result concerns a second-order stochastic dominance change. Suppose the expected value of agent $i$'s type $t_i$, conditional on him being in favor, increases. Then Proposition \ref{prop:opt_weights} implies that his optimal weight $\omega_i^+$ increases as well.

\section{BIC-EPIC equivalence}\label{sec_bicepic}

A voting-with-evidence mechanism is not only Bayesian incentive compatible, but it also  satisfies the stronger notion of ex-post incentive compatibility (see Remark \ref{rem_expost_cutoff}). This robustness of the voting-with-evidence mechanism is a desirable property of any mechanism that one  wish to use in real-life applications because optimal strategies are independent of beliefs and information structure. Reducing the number of assumptions about common knowledge and weakening the informational requirements places the theoretical analysis underpinning the design on firmer ground~(\citeasnoun{wilson1987} and~\citeasnoun{bergemann2005}).  

Because the optimal mechanism is ex-post incentive compatible we conclude that the principal cannot gain by weakening the incentive constraints. A natural question to ask is why the principal cannot save on verification costs by implementing the optimal mechanism in Bayesian equilibrium instead of ex-post equilibrium. We show that the answer lies in a general equivalence between Bayesian and ex-post incentive compatible mechanisms.  For \emph{every} BIC mechanism, there exists an ex-post incentive compatible mechanism that induces the same interim expected decision and verification rules; since the interim expected decision and verification rules determine the expected utility of the principal, this implies that an ex-post incentive compatible mechanism is optimal within the whole class of BIC mechanisms. 

  Recall that a mechanism $(d,a)$ is BIC if and only if, for all $i$ and $t_i$,
  \begin{align}
     \inf_{t_i' \in T_i^+} \E_{t_{-i},s}[d(t_i',t_{-i},s)] &\ge \E_{t_{-i},s}[d(t_i,t_{-i},s) [1-a_i(t_i,t_{-i},s)]]  \hspace{.5cm}\text{ and } \label{eq:bic1}\\
      \sup_{t_i' \in T_i^-} \E_{t_{-i},s}[d(t_i',t_{-i},s)] &\le \E_{t_{-i},s}[d(t_i,t_{-i},s) [1-a_i(t_i,t_{-i},s)] + a_i(t_i,t_{-i},s)].\label{eq:bic2}
  \end{align}
  
  Analogously, a mechanism $(d,a)$ is ex-post incentive compatible (EPIC) if and only if, for all $i, t_i$ and $t_{-i}$,
  \begin{align}
     \inf_{t_i' \in T_i^+} \E_{s}[d(t_i',t_{-i},s)] &\ge \E_{s}[d(t_i,t_{-i},s) [1-a_i(t_i,t_{-i},s)]]  \hspace{.5cm}\text{ and } \label{eq:epic1} \\
      \sup_{t_i' \in T_i^-} \E_{s}[d(t_i',t_{-i},s)] &\le \E_{s}[d(t_i,t_{-i},s) [1-a_i(t_i,t_{-i},s)] + a_i(t_i,t_{-i},s)]. \label{eq:epic2}
   \end{align} 

Not every BIC mechanism is EPIC. More important, not every decision rule that can be implemented in a Bayesian equilibrium can be implemented in an ex-post equilibrium with the same verification costs, as the following example illustrates. 

\newcommand{\mechaxis}{
  \draw[gray] (0,0) rectangle (-1,1);
  \draw[axis,->] (-1,0) -- (0.08,0) node[right] {$t_1$};
  \draw[axis,->] (-1,0) -- (-1,1.08) node[above] {$t_2$};
  \draw[dashed] (0,1/3) -- (-1,1/3);
  \draw[dashed] (0,2/3) -- (-1,2/3);
  \draw[dashed] (-1/3,0) -- (-1/3,1);
  \draw[dashed] (-2/3,0) -- (-2/3,1);
  \draw[thick,blue] (-.98,1.08) -- (0.2,1.08);
  \draw[thick,blue] (.08,1.2) -- (.08,0.02);

      \node[sm,blue] at (-1/6,1.18) {$0.5$};
      \node[sm,blue] at (-3/6,1.18) {$0.7$};
      \node[sm,blue] at (-5/6,1.18) {$0.2$};
      \node[sm,blue] at (.2,1/6) {$0.8$};
      \node[sm,blue] at (.2,3/6) {$0.2$};
      \node[sm,blue] at (.2,5/6) {$0.4$};
      }
\newcommand{\mechaxisnew}{
  \draw[gray] (0,0) rectangle (-1,1);
  \draw[axis,->] (-1,0) -- (0.08,0) node[right] {$t_1$};
  \draw[axis,->] (-1,0) -- (-1,1.08) node[above] {$t_2$};
  \draw[dashed] (0,1/3) -- (-1,1/3);
  \draw[dashed] (0,2/3) -- (-1,2/3);
  \draw[dashed] (-1/3,0) -- (-1/3,1);
  \draw[dashed] (-2/3,0) -- (-2/3,1);
  \draw[thick,blue] (.08,1.2) -- (.08,0.02);
      }

\begin{figure}\centering%
  \subfloat[Decision rule $d$ and its marginals in blue.]{%
    \label{fig:BICrule}%
    \begin{tikzpicture}[baseline,x=4.5cm,y=4.5cm]
    \mechaxis

      \fill[pattern=north west lines] (-2/3,1/3) rectangle (-3/3,2/3);
      \fill[pattern=north west lines] (-1/3,1/3) rectangle (-2/3,2/3);
      \fill[pattern=north west lines] (0,2/3) rectangle (-1/3,3/3);

      \node[sm] at (-1/6,5/6) {$0$};
      \node[sm] at (-3/6,5/6) {$1$};
      \node[sm] at (-5/6,5/6) {$0.2$};
      \node[sm] at (-1/6,3/6) {$0.5$};
      \node[sm] at (-3/6,3/6) {$0.1$};
      \node[sm] at (-5/6,3/6) {$0$};
      \node[sm] at (-1/6,1/6) {$1$};
      \node[sm] at (-3/6,1/6) {$1$};
      \node[sm] at (-5/6,1/6) {$0.4$};
    \end{tikzpicture}%
  } \hspace{3cm}
    \subfloat[Verification probabilities that are necessary for EPIC.]{%
    \label{fig:EPIC_ver}%
    \begin{tikzpicture}[baseline,x=4.5cm,y=4.5cm]
      \mechaxisnew

      \fill[pattern=north west lines] (-2/3,1/3) rectangle (-3/3,2/3);
      \fill[pattern=north west lines] (-1/3,1/3) rectangle (-2/3,2/3);
      \fill[pattern=north west lines] (0,1/3) rectangle (-1/3,2/3);

      \node[sm] at (-1/6,5/6) {$0$};
      \node[sm] at (-3/6,5/6) {$0.9$};
      \node[sm] at (-5/6,5/6) {$0.2$};
      \node[sm] at (-1/6,3/6) {$0.5$};
      \node[sm] at (-3/6,3/6) {$0$};
      \node[sm] at (-5/6,3/6) {$0$};
      \node[sm] at (-1/6,1/6) {$1$};
      \node[sm] at (-3/6,1/6) {$0.9$};
      \node[sm] at (-5/6,1/6) {$0.4$};
      \node[sm,blue] at (.2,1/6) {$\frac{2.3}{3}$};
      \node[sm,blue] at (.2,3/6) {$\frac{0.5}{3}$};
      \node[sm,blue] at (.2,5/6) {$\frac{1.1}{3}$};

    \end{tikzpicture}%
  }\qquad%
  \caption{Failure of a naive BIC-EPIC equivalence.}%
  \label{fig:BICEPIC}


\end{figure}
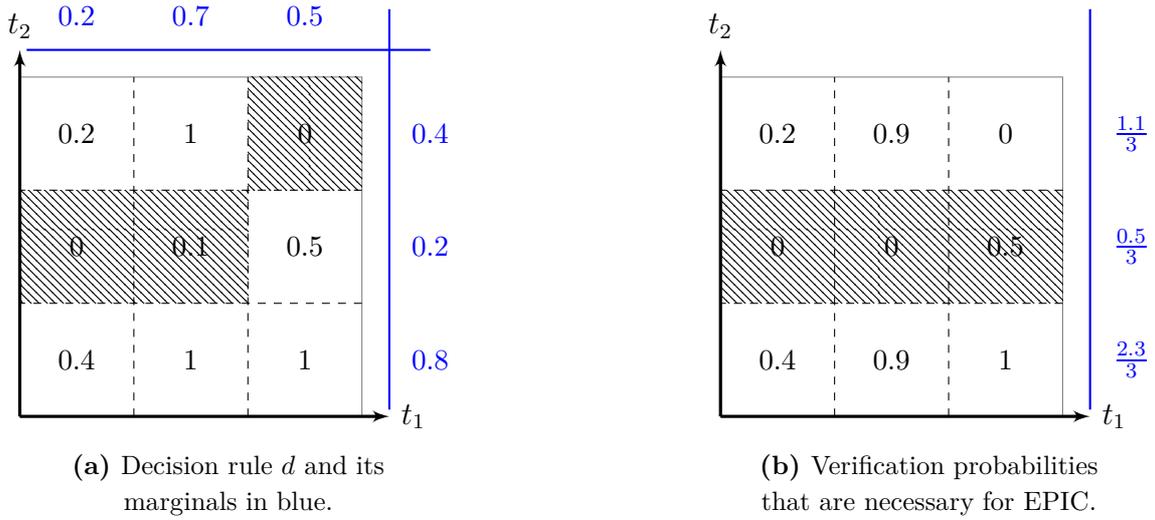
   \begin{example}\rm\label{ex_bic_epic_naive}
  Suppose that  $\I=\{1,2\}$ and that agent 2 is always in favor of the new policy. Each type profile is equally likely and the decision rule $d$ is shown in Figure \ref{fig:BICrule}. The shaded areas indicate type profiles that induce the lowest probabilities of accepting the new policy for agent 2. We focus on incentive constraints for agent 2.

  Lemma~\ref{lemma_bic} shows that it is sufficient to ensure incentive compatibility for the ``worst-off'' types, which are the intermediate types in this example. Since the intermediate types are the worst-off, they never need to be verified. If high (low) types are verified with probability 0.2 (0.6), then the Bayesian incentive constraints for the worst-off types are exactly binding.
  If we instead want to implement the decision rule $d$ in an ex-post equilibrium, the cost of verification increases. For example, intermediate types must be verified with probability 0.5 if agent 1's type is high. In expectation, agent 2 must be verified with probability $\frac{0.5}{3}$ if he has an intermediate type, with probability $\frac{1.1}{3}$ if he has a high type, and with probability $\frac{2.3}{3}$ if he has a low type (the verification probabilities for each profile of reports are given in Figure~\ref{fig:EPIC_ver}).
\end{example}

As Example~\ref{ex_bic_epic_naive} above illustrates, we cannot simply take a BIC mechanism, maintain the same decision rule, and expect that the mechanism will also be EPIC without increasing the verification costs. This is in line what should  be expected since for a mechanism to be EPIC, incentive constraints must hold pointwise and not only in expectation. 
The reason for this is that in general the left-hand side of \eqref{eq:bic1} is greater than the expected value of the left-hand side of \eqref{eq:epic1}; that is, $\inf_{t_i' \in T_i^+} \E_{t_{-i},s}[d(t_i',t_{-i},s)]$ is generally larger than $\E_{t_{-i}}\big[\inf_{t_i' \in T_i^+} \E_{s}[d(t_i',t_{-i},s)]\big]$. 
A decision rule can be implemented in ex-post equilibrium at the same costs as in Bayesian equilibrium if and only if the expectation operator commutes with the infimum/supremum operator, which is a strong requirement.
However, it turns out that for every function, there exists another function that induces the same marginals and for which the expectation operator commutes with the infimum/supremum operator. We will use this result to establish an equivalence between BIC and EPIC mechanisms.

\begin{theorem} \label{thm:commute}
Let $A = \bigtimes_i A_i \subseteq \mathbb{R}^I$, let $t_i$ be independently distributed with an absolutely continuous distribution function $F_i$, and let $g: A \rightarrow [0,1]$ be a measurable function. Then there exists a function $\hat{g}: A \rightarrow [0,1]$ with the same marginals, \iec, for all $i$, $\E_{t_{-i}}[g(\cdot,t_{-i})]  = \E_{t_{-i}}[\hat{g}(\cdot,t_{-i})] $ almost everywhere, such that for all $B \subseteq A_i$,
  \begin{align*}
    \inf_{t_i \in B}\E_{t_{-i}}[\hat{g}(t_i,t_{-i})] &= \E_{t_{-i}}[\inf_{t_i \in B}\hat{g}(t_i,t_{-i})] \text{ and }\\
    \sup_{t_i \in B}\E_{t_{-i}}[\hat{g}(t_i,t_{-i})] &= \E_{t_{-i}}[\sup_{t_i \in B}\hat{g}(t_i,t_{-i})].
  \end{align*}
\end{theorem}

We will illustrate the idea behind the proof of Theorem \ref{thm:commute} by assuming that $A$ is finite. The argument in our proof uses Theorem 6 in \citeasnoun{gutmann91}. This theorem shows that for any matrix with elements between 0 and 1 and with increasing row and column sums, there exists another matrix consisting of elements between 0 and 1 with the same row and column sums, and  whose elements are increasing in each row and column. To use this result, we reorder $A$ such that the marginals of $g$ are weakly increasing. Then, Theorem 6 in \citeasnoun{gutmann91} implies that there exists a function $\hat{g}$ that induces the same marginals and is pointwise increasing. For this function, there is an argument $t_i$ for each $i$ that independent of $t_{-i}$ minimizes $\hat{g}(\cdot, t_{-i})$. This implies that the expectation operator commutes with the infimum operator, i.e.,  $\E_{t_{-i}}[ \inf_{t_i \in A} \hat{g}(t_i,t_{-i})] = \inf_{t_i \in A} \E_{t_{-i}}[\hat{g}(t_i,t_{-i})]$. 
This basic idea sketched above is extended via an approximation argument to a complete proof in Appendix~\ref{proof_bicepic}.

\medskip
 
Building on Theorem \ref{thm:commute}, we can establish an equivalence between BIC and EPIC mechanisms.
To define this equivalence formally, we call $\E_{t_{-i}}[d(t_i,t_{-i})]$ the \textit{interim decision rule} and $\E_{t_{-i}}[a_i(t_i,t_{-i})]$ the \textit{interim verification rules} of a mechanism $(d,a)$.

  \begin{definition}\rm
    Two mechanisms $(d,a)$ and $(\hat{d},\hat{a})$ are \emph{equivalent} if they induce the same interim decision and verification rules almost everywhere.
  \end{definition}
  
Now we can state the equivalence between BIC and EPIC mechanisms. 

  \begin{theorem}
    For any BIC mechanism $(d,a)$, there exists an equivalent EPIC mechanism $(\hat{d}, \hat{a})$. \label{prop:bic-epic}
  \end{theorem}

There are two steps in the construction of an equivalent EPIC mechanism $(\hat{d},\hat{a})$. In the first step, we use Theorem \ref{thm:commute} to obtain a decision rule $\hat{d}$ with the same interim decisions as $d$ and such that for $\hat{d}$ the expectation operator commutes with the infimum/supremum. This implies that the left-hand sides of \eqref{eq:bic1} and \eqref{eq:bic2}, respectively  are equal to the expected values of the left-hand sides of \eqref{eq:epic1} and \eqref{eq:epic2}, respectively. In the second step, we construct a verification rule $\hat{a}$ such that all incentive constraints hold as equalities for $(\hat{d},\hat{a})$. By potentially adding some verification, we obtain a verification rule $\hat{a}$ with the same interim verification rule as $a$. Thus, we have constructed an equivalent EPIC mechanism $(\hat{d}, \hat{a})$ from the BIC mechanism $(d,a)$. 

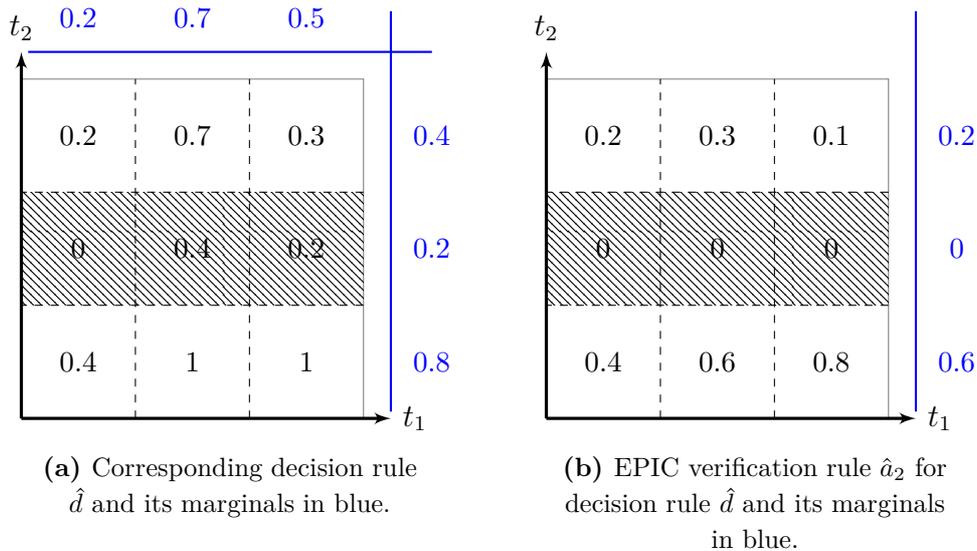
\begin{figure}\centering
    \subfloat[Corresponding decision rule $\hat{d}$ and its marginals in blue.]{%
    \label{fig:EPIC_decision_rule}%
    \begin{tikzpicture}[baseline,x=4.5cm,y=4.5cm]
      \mechaxis

      \fill[pattern=north west lines] (-2/3,1/3) rectangle (-3/3,2/3);
      \fill[pattern=north west lines] (-1/3,1/3) rectangle (-2/3,2/3);
      \fill[pattern=north west lines] (0,1/3) rectangle (-1/3,2/3);

      \node[sm] at (-1/6,5/6) {$0.3$};
      \node[sm] at (-3/6,5/6) {$0.7$};
      \node[sm] at (-5/6,5/6) {$0.2$};
      \node[sm] at (-1/6,3/6) {$0.2$};
      \node[sm] at (-3/6,3/6) {$0.4$};
      \node[sm] at (-5/6,3/6) {$0$};
      \node[sm] at (-1/6,1/6) {$1$};
      \node[sm] at (-3/6,1/6) {$1$};
      \node[sm] at (-5/6,1/6) {$0.4$};

    \end{tikzpicture}%
  }\qquad%
    \subfloat[EPIC verification rule $\hat{a}_2$ for decision rule $\hat{d}$ and its marginals in blue.]{%
    \label{fig:EPICrule}%
    \begin{tikzpicture}[baseline,x=4.5cm,y=4.5cm]
      \mechaxisnew

      \fill[pattern=north west lines] (-2/3,1/3) rectangle (-3/3,2/3);
      \fill[pattern=north west lines] (-1/3,1/3) rectangle (-2/3,2/3);
      \fill[pattern=north west lines] (0,1/3) rectangle (-1/3,2/3);

      \node[sm] at (-1/6,5/6) {$0.1$};
      \node[sm] at (-3/6,5/6) {$0.3$};
      \node[sm] at (-5/6,5/6) {$0.2$};
      \node[sm] at (-1/6,3/6) {$0$};
      \node[sm] at (-3/6,3/6) {$0$};
      \node[sm] at (-5/6,3/6) {$0$};
      \node[sm] at (-1/6,1/6) {$0.8$};
      \node[sm] at (-3/6,1/6) {$0.6$};
      \node[sm] at (-5/6,1/6) {$0.4$};
      \node[sm,blue] at (.2,1/6) {0.6};
      \node[sm,blue] at (.2,3/6) {0};
      \node[sm,blue] at (.2,5/6) {0.2};
    \end{tikzpicture}%
  }\qquad%
  \caption{Illustration of the BIC-EPIC equivalence.}%
  \label{fig:EPIC_decision_rule}
\end{figure}
  \addtocounter{example}{-1}
  \begin{example}\rm[ctd]
    Figure \ref{fig:EPICrule} depicts the decision rule $\hat{d}$, which has the same marginals as $d$.  Note that intermediate types of agent 2 always induce the lowest probability of accepting the proposal, independent of the type of agent 1. This implies that the expected value of the infimum equals the infimum of the expected value, that is, 
    \[\inf_{t_2} \E_{t_1}[\hat{d}(t)] = \E_{t_1}[\inf_{t_2} \hat{d}(t)].\]
    Figure \ref{fig:EPICrule} shows a verification rule $\hat{a}$ such that $(\hat{d},\hat{a})$ is EPIC. The expected verification probabilities are the same as those necessary for implementation in Bayesian equilibrium.
  \end{example}

The economic mechanisms behind our equivalence are different from those underlying the BIC-DIC equivalence in a standard social choice setting with transfers (with linear utilities and one-dimensional, private types~\cite{gershkov13}). In the standard setting, an allocation rule can be implemented with appropriate transfers in Bayesian equilibrium if and only if its marginals are increasing and in dominant strategies if and only if it is pointwise increasing. In contrast, monotonicity is neither necessary nor sufficient for implementability in our model. 

Note that there is no equivalence between Bayesian and dominant-strategy incentive compatible mechanisms in our setting, as the following example illustrates. The lack of private goods to punish agents if there are multiple deviators implies that agents care whether the other agents are truthful. 

  \begin{example}\rm\label{ex:not_bic_dic}
    Suppose that  $\mathcal{I}=\{1,2,3\}$, verification costs are 0 for each agent, and $T_i^+ = \{t_i | t_i \ge 0\}$ and $T_i^- = \{t_i | t_i < 0\}$. Consider the voting-with-evidence mechanism with cutoffs $\omega_i^+ +c_i=1$ and $\omega_i^- -c_i=-1$ for all $i$. Let $ t = (-5,2,2)$. Given truthful reporting, the voting-with-evidence mechanism   specifies that $d(t)=0$. Suppose that agent 2 deviates from truth-telling and instead reports being of type $t_2'=6$. Now he is decisive, and the principal verifies him. After observing the true types $(-5,2,2)$, the principal has to punish the lie by agent 2 and maintain the status quo to induce truthful reporting. However, this creates an incentive for agent 3 to misreport. He could report $t_3'= 6$, and then no agent is decisive; hence, no one is verified, and the voting-with-evidence mechanism specifies  that $d(t_1,t_2',t_3')=1$. The voting-with-evidence mechanism is therefore not dominant-strategy incentive compatible, no matter how we specify the mechanism off-equilibrium.
  \end{example}

The equivalence between Bayesian and ex-post incentive compatible mechanisms can be established in other models without money but with verification. We believe that the tools we used in this paper can prove useful in similar settings with verification. In fact, we can use arguments paralleling those used to prove Theorem \ref{thm:commute} (but using Theorem 1 in \citeasnoun{gershkov13} instead of the result by \citeasnoun{gutmann91}) to show that there is an equivalence of Bayesian and dominant-strategy incentive compatible mechanisms in BDL.

\section{Imperfect Verification and  Robustness}\label{sec_rob_ext}

In this section we  discuss the robustness of our results from various angles. In the first part, we relax the assumption of perfect verification. In the second part, we discuss briefly type-dependent costs of verification,  interdependent preferences, a continuous decision on the level of the public good, and limited commitment.

\subsection{Imperfect verification} \label{sec:imperfect}  
Thus far, we have assumed that the verification technology works perfectly, that is, whenever the principal audits an agent, he will learn the true type with probability one. We now explore the extent to which the above results are robust to imperfect verification. We will study a reduced form model and assume that in the event of an audit of agent $i$, the verification technology reveals the true type of agent $i$ only with probability $p$, and with probability $1-p$, the technology fails, in which case the output of the technology equals the report by the agent. Consequently, if the verification output differs from the reported type the principal knows that the agent lied. However, if the output of the verification technology coincides with the reported type the principal only knows that the agent was truthful or that the verification technology failed, but not which of these two cases applies. Moreover, we assume that multiple verifications of the same agent reveal no additional information.

To find the optimal mechanism we first characterize Bayesian incentive compatibility in this new setting with imperfect verification. Similar to the case of perfect verification the key incentive constraints are those for the worst-off types. The additional uncertainty of whether the verification technology managed to detect a lie implies that the worst-off type must get a higher expected probability of getting the preferred alternative. 

\begin{lemma}\label{lemma_bic_impp}
A mechanism $(d,a)$ is Bayesian incentive compatible if and only if, for all $i\in \I$ and all $t_i\in T_i$,
\begin{align*}
     \inf_{t_i'\in T_i^+} \E_{t_{-i},s}[d(t_i',t_{-i},s)] &\ge \E_{t_{-i},s}[d(t_i,t_{-i},s) [1-p\cdot a_i(t_i,t_{-i},s)]]  \hspace{.5cm} \\
      \sup_{t_i'\in T_i^-} \E_{t_{-i},s}[d(t_i',t_{-i},s)] &\le \E_{t_{-i},s}[d(t_i,t_{-i},s) [1-p\cdot a_i(t_i,t_{-i},s)] + p \cdot a_i(t_i,t_{-i},s)].
\end{align*}
\end{lemma}
\begin{proof}
The proof is analogous to the proof of Lemma \ref{lemma_bic}.
\end{proof}
\noindent
The imperfectness of the verification technology implies  that it is harder to satisfy the incentive constraints. Moreover, there is an upper bound on how much influence an agent can have in expectation. Since, by feasibility $a_i(t,s)\le 1$ and using  Lemma \ref{lemma_bic_impp} we get that any Bayesian incentive compatible mechanism satisfies
\begin{align}
\forall t_i \in T_i^+: \ \E_{t_{-i},s}[d(t_i,t_{-i},s)] &\le \frac{1}{1-p} \inf_{t_i'\in T_i^+} \E_{t_{-i},s}[d(t_i',t_{-i},s)]  \label{eq:bic_imp1}\\
\forall t_i \in T_i^-: \ \E_{t_{-i},s}[d(t_i,t_{-i},s)] &\ge \frac{1}{1-p} \left[\sup_{t_i'\in T_i^-} \E_{t_{-i},s}[d(t_i',t_{-i},s)] - p\right]. \label{eq:bic_imp2}
\end{align}
This adds an additional constraint to the relaxed problem that essentially restricts the maximal influence an agent could have on the decision rule in any incentive compatible mechanism. The higher the probability of failure $1-p$ of the verification technology the tighter the bound is and the less influence an agent can have. 

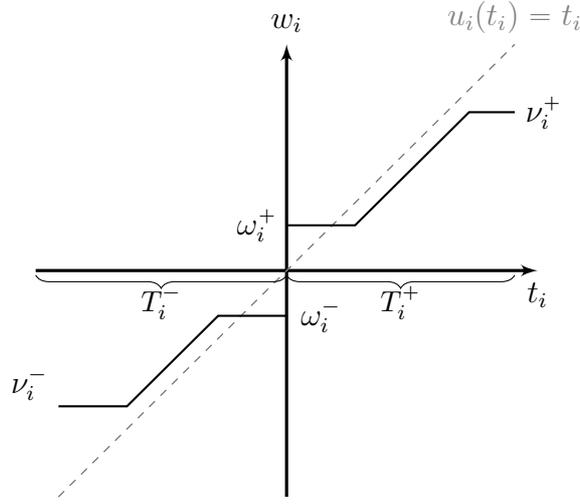
\begin{figure}[htbp!]%
\begin{center}
  \begin{tikzpicture}[decoration={brace},baseline,x=3cm,y=3cm]
      \draw[axis,->] (-1.1,0) -- (1.1,0) node[below] {$t_i$};
      \draw[axis,->] (0,-1) -- (0,1) node[above] {$w_i$};
       \draw[dashed,gray] (-1,-1) -- (1,1) node[above]{$u_i(t_i)=t_i$};
      \draw[dashed,gray] (-1,-1) -- (1,1);
      \draw[thick] (-1,-.6) -- (-.7,-.6) -- (-.3,-.2) -- (0,-.2);
      \draw[thick]  (0,.2) -- (.3,.2) -- (.8,.7) -- (1,.7);
      \draw (-1,-.5) node[left] {$\nu_i^{-}$};
      \draw (1,.7) node[right] {$\nu_i^{+}$};
      \draw (.15,-.2) node {$\omega_i^-$};
      \draw (-.13,.2) node {$\omega_i^+$};
      \draw[decorate,decoration={amplitude=5pt,mirror}] (0,-.02) -- (1,-.02) node[midway,below] {\small $T_i^+$};
      \draw[decorate,decoration={amplitude=5pt}] (0,-.02) -- (-1.1,-.02) node[midway,below] {\small $T_i^-$};
  \end{tikzpicture}
\end{center}
\caption{Example illustrating weights for imperfect verification with utility $u_i(t_i)= t_i$.}%
\label{fig:imp}%
\end{figure}

\begin{theorem}\label{thm:imp}
With imperfect verification as described above, an optimal mechanism sets $d(t) = 1$ if and only if  $\sum_i w_i(t_i) > 0$, where 
\begin{align*}
  w_i(t_i) = \begin{cases}
    \omega_i^{+} \hspace{.3cm} &\text{ if } t_i \in T_i^+ \text{ and }  t_i \le \omega_i^{+} + \frac{c_i}{p} \\
    \omega_i^{-} &\text{ if } t_i \in T_i^- \text{ and } t_i \ge \omega_i^{-}-\frac{c_i}{p} \\
    t_i - \frac{c_i}{p} &\text{ if } t_i \in T_i^+ \text{ and } \nu_i^+ > t_i- \frac{c_i}{p}   > \omega_i^{+}\\
    t_i + \frac{c_i}{p} &\text{ if } t_i \in T_i^- \text{ and } \nu_i^- < t_i +\frac{c_i}{p}  < \omega_i^{-} \\
    \nu_i^+  &\text{ if } t_i \in T_i^+ \text{ and } t_i \ge \nu_i^+ +\frac{c_i}{p}   \\
    \nu_i^-  &\text{ if } t_i \in T_i^- \text{ and } t_i \le \nu_i^- - \frac{c_i}{p} 
  \end{cases}
\end{align*}
for some constants $\{\omega_i^+,\omega_i^-,\nu_i^+,\nu_i^-\}$ satisfying $\omega_i^-\le \omega_i^+$.
\end{theorem}

Compared to the optimal mechanism in the benchmark model with perfect verification, the optimal mechanism with imperfect verification can feature additional bunching regions at the extremes (see Figure \ref{fig:imp}). The reason is that any incentive compatible mechanism must restrict the maximal weight of an agent compared to the worst-off types. If a worst-off type could misreport his type and thereby increase the probability of getting his preferred outcome by too much, this misreport would be a profitable deviation even if this agent was always verified, simply because the verification technology sometimes fails to detect the lie. Therefore, any incentive compatible mechanism must cap the maximal weight an agent could get, inducing bunching for the extreme types.\footnote{This is reminiscent of the optimal mechanism in \citeasnoun{mylovanov17}, who study the optimal allocation of a prize when the winner is subject to a limited penalty if he makes a false claim. In their model, the limit on the penalty similarly requires that agents with the highest possible type are merely short-listed and will not win the prize with certainty.} 

In contrast to the optimal mechanism with perfect verification it is not enough to only verify decisive agents. Clearly, one should only verify an agent that gets his preferred outcome, but to induce truth-telling as a Bayes-Nash equilibrium one sometimes needs to verify agents who are not decisive. Suppose only  decisive agents were verified. Clearly, agents with worst-off types would have an incentive to overstate their types because this could never hurt them (they are only verified if they are decisive, in which case the penalty is the outcome they would have obtained under truth-telling), and benefits them whenever the verification technology fails to detect their lie. Thus, agents sometimes need to be verified even so they are not decisive; by doing this sufficiently often we can ensure that the mechanism is BIC. There are several ex-post auditing rules that can make the optimal mechanism BIC, and we have not specified exactly which agents are going to be verified for a given realization of reports. We establish the existence of a feasible auditing rule in Lemma~\ref{lemma:imp_feas}.
This reasoning also implies that the optimal mechanism is not ex-post incentive compatible if the verification technology is imperfect.

\begin{remark}\rm
The additional bunching regions are the main qualitative difference of the optimal decision rule compared to the model with perfect verification. However, in many settings this difference will not even arise: if an optimal decision rule $d$ (as described in Theorem \ref{thm:imp}) satisfies, for each $i$,
\begin{align*}
(1-p) \sup_{t_i\in T_i^+} \E_{\ti}d(t_i,\ti) &< \inf_{t_i\in T_i^+} \E_{\ti}d(t_i,\ti) \text{ and }\\
(1-p) \inf_{t_i\in T_i^-} \E_{\ti}d(t_i,\ti) &> \sup_{t_i\in T_i^-} \E_{\ti}d(t_i,\ti),
\end{align*}
then $\nu_i^+ = \infty$ and $\nu_i^-=-\infty$ (see the proof of Lemma \ref{lemma:imperfect_opt}). Therefore, the weight function looks qualitatively as in the case of perfect verification. For example, if $p> \frac{1}{2}$ then the above conditions are always satisfied for a symmetric mechanism ($d$ is symmetric around 0) in a symmetric environment ($f_i(t_i)=f_i(-t_i)$ and $T_i^+=-T_i^-$), because $\inf_{t_i\in T_i^+} \E_{\ti}d(t_i,\ti)\ge \frac{1}{2}$ and $\sup_{t_i\in T_i^-} \E_{\ti}d(t_i,\ti)\le \frac{1}{2}$. 
\end{remark}

\subsection{Robustness}\label{sec_robust}
In the remainder we are going to keep the assumption of perfect verification and change some of our other assumptions to inquire which features of our analysis are robust.
 
\paragraph{Type-dependent cost function}
In our benchmark model we assume that the cost of verifying an agent only depends on the agent's identity and not on his true type. Alternatively, one could argue that it's more expensive to audit an agent that claims to have a large type and provides extensive documentation substantiating his claim. 
Similarly, one could argue that it is easier, and therefore cheaper, to verify an agent with a low type. Here, we explain how our conclusions are altered if we allow for the audit cost to depend on  the true type.  Let $c_i(t_i)$ denote the audit cost for verifying agent $i$ if his true type is $t_i$. 

We observe first that the revelation principle still applies, and we can restrict attention to Bayesian incentive compatible direct mechanisms. Also, this change only affects the principal's utility and we can therefore use the characterization of Bayesian incentive compatibility as before. On the equilibrium path, the principal will only verify agents that are truthful, so the cost of verification is $c_i(t_i)$.
To simplify the discussion, we assume that the net type $t_i-c_i(t_i)$ is increasing in $t_i$.\footnote{If the net type $t_i-c_i(t_i)$ is not increasing similar arguments can be applied after the types have been reordered.}  Using the same arguments as in our benchmark model we can conclude that the optimal mechanism uses a weighting rule as in a voting-with-evidence mechanism, except that the weight of a report outside of the bunching region is now $t_i-c_i(t_i)$ instead of $t_i-c_i$ (respectively, $t_i+c_i(t_i)$). The part of the weighting function outside the bunching region is therefore no longer a straight line with a slope of one but a potentially nonlinear increasing function instead. Other than that the optimal mechanism is like a voting-with-evidence mechanism: the project is implemented if the sum of the weighted reports is positive and an agent is verified if and only if he is decisive.

\paragraph{Choosing the level of the public good} 
In our benchmark model we assume the principal takes a binary action by deciding whether or not to implement a public project. We relax this assumption here and analyze a principal who decides on the quantity $d\in [0,1]$ of a public good and assume that the principal pays a cost of $C(d)$ for providing the public good at level $d$. Assume that the cost function $C:[0,1]\rightarrow \mathbb{R}$ is continuously differentiable, increasing, convex, and satisfies $C'(0)=0$ and $\lim_{d\rightarrow 1}C(d)=+ \infty$.  All agents have preferences of the form  $u(t_i,d)= t_i\cdot d$  and $t_i\in[0,1]$, i.e.,  agents always prefers more of the public good to less. The objective function of the principal is 
\begin{align}\label{eq:mult}
 \E_t\big[\sum_i [d(t)t_i  - a_i(t)c_i] -C(d(t))\big]
\end{align}
since he incurs the additional costs $C(d)$ of providing the public good.

We begin the discussion with the simplest setting of having only one agent. 
It follows again from Lemma 1 that any incentive compatible mechanism satisfies $ \inf_{t'} d(t') \ge d(t)(1-a(t)) $ and, since audits are costly, it is optimal to choose the verification rule such that this holds as an equality. Plugging this into the objective function, we get that the principal maximizes $\E_t\big[d(t)t  -(1-\frac{\inf_{t'} d(t') }{d(t)}) c -C(d(t))\big]$.\footnote{Observe that, for an optimal decision rule, $\inf_{t'} d(t')>0$. Suppose instead $\inf_{t'} d(t')=0$. Given $\varepsilon>0$, let $\delta(\varepsilon)=\Prob(\{t|d(t)\le \varepsilon\})$. If there is $\varepsilon > 0$ such that $\delta(\varepsilon)<1$, we can change the decision rule such that $\inf_{t'}d(t')=\varepsilon$ by only changing the decision for types in $\{t|d(t)<\varepsilon\}$. This change will increase the cost of public good provision by at most $\delta(\varepsilon) \varepsilon C'(\varepsilon)$, but will decrease the cost of verifications by at least $(1- \delta(\varepsilon)) \varepsilon c$. Therefore, for $\varepsilon$ small enough this increases the principal's expected payoff. On the other hand, if $\delta(\varepsilon)=1$ for all $\varepsilon>0$ then $d(t)=0$ for almost every $t$. Changing the decision rule to $d(t)=\varepsilon$  increases the cost of public good provision by at most $\varepsilon C'(\varepsilon)$ and increases the expected welfare of the principal by $\varepsilon \E[t]$. Since $C$ is continuously differentiable and $C'(0)=0$ we can therefore choose $\varepsilon$ small enough such that this change increases the principal's expected welfare. We conclude that in any 
optimal mechanism $\inf_{t'} d(t')>0$.}

The optimal decision rule $d$ must therefore satisfy, for almost every $t$ such that $d(t)>\inf_{t'} d(t')$, the following first-order condition:
\begin{align}\label{eq:foc} 
 t- c  \frac{\inf_{t'} d(t') }{d^2(t)} - C'(d(t))=0.
 \end{align} 
Therefore, as before there is a bunching region, and outside the bunching region we have downward-distortions, i.e., too little public good is provided and this distortion is increasing in the verification cost. 
Note that in contrast to the previous analysis, where the quantity was either $0$ or $1$, optimal audits are now stochastic (as they satisfy $a(t)=1-\frac{\inf_{t'} d(t') }{d(t)}$).
Intuition suggests that these conclusions for one agent carry over to the case with multiple agents if we impose ex-post incentive compatibility instead of Bayesian incentive compatibility. 

Let us now look briefly at the case with several agents and Bayesian incentive constraints. The characterization of Bayesian incentive compatibility in Lemma~\ref{lemma_bic} continues to hold in this setting. Although incentive constraints remain tractable, solving the principal's problem turns out to be less tractable.
The principal's optimization problem is to maximize \eqref{eq:mult} subject to
\begin{align}\label{eq:bicmult}
\E_{t_{-i}} [d(t_i,t_{-i})]\ge \inf_{t_i'}\E_{t_{-i}} [d(t_i',t_{-i})] \ge \E_{t_{-i}} \big[d(t_i,t_{-i})(1-a_i(t_i,t_{-i}))\big].
\end{align}
Consider first how to construct the optimal audit rule for a given decision rule $d$. Again, the optimal audit rule will satisfy the second inequality in \eqref{eq:bicmult} as an equality. To achieve this in the most cost efficient way, we set, for each $i$ and $t_i$,  $a_i(t_i,t_{-i})=1$ for those $t_{-i}$ such that $d(t_i,t_{-i})$ is largest until the second inequality in \eqref{eq:bicmult} binds. Thus, the optimal verification rule is deterministic. This is in contrast to the analysis above for the case of one agent.
It also implies that there is no simple way to compute the verification costs necessary to implement a given decision rule and we cannot formulate the problem in a simple way with the decision rule being the only choice variable. Because of this, a complete analysis of the optimal decision rule in this case is beyond the scope of our paper.

\paragraph{Interdependent preferences}
Independent private values allow for a simple characterization of incentive compatibility: a mechanism is Bayesian incentive compatible if and only if it is Bayesian incentive compatible for the worst-off types. This observation does not carry over to models with interdependent preferences. While a complete analysis of this case is beyond the scope of this paper, we discuss below the incentives to misreport in a voting-with-evidence mechanism and possible improvements of this mechanism when preferences are interdependent. To fix ideas, for each $i\in \I$, suppose $T_i=[-1,1]$ and agent $i$'s utility is given by $u_i(t)= t_i +\alpha\sum_{j\neq i} t_j$ if the policy is implemented and the type profile is $t$, where $\alpha$ satisfies $0<\alpha< 1$. For our discussion below consider a fixed voting-with-evidence mechanism.

If the level of interdependence, $\alpha$, is high enough then incentive constraints in the voting-with-evidence mechanism are not binding.
Recall that in our benchmark model reports are verified exactly to make worst-off types indifferent between lying and being truthful.
If preferences are sufficiently interdependent, an agent with a small positive type might not want to deviate and send a large report even without verifications: his utility is mainly determined by other agents' types, and claiming a high type might lead to implementation of the new policy although all others have negative types. This implies that one can reduce the verification probability of large reports without creating any incentives to misreport. However, one cannot lower the verification probability
all the way to zero, as otherwise intermediate types will have an incentive to send high reports.
Which incentive constraints will be binding in the optimal mechanism therefore depends on the details of preferences and type distributions. This implies that it is difficult to find the  optimal mechanism. But the arguments so far suggest that one way to improve upon a voting-with-evidence mechanism might be to reduce the verification probabilities for high reports, at least if $\alpha$, the degree of preference interdependence, is sufficiently large. 

For moderate degrees of interdependence, $\alpha$, all types above a threshold will prefer the new policy no matter what the types of all others are since they are only moderately affected by others' types. For these types, incentives are as in our benchmark model since these types will send a report to maximize the probability that the new policy is implemented.
Furthermore, for small enough $\alpha$ this is even true for some types in the bunching region of the voting-with-evidence mechanism. Since the worst-off types were used to determine the verification probabilities, we cannot reduce the verification probabilities of voting-with-evidence at all for small degrees of interdependence.

This suggests that there are only limited ways to improve upon voting-with-evidence if $\alpha$ is small. One particularly simple way to improve upon a voting-with-evidence mechanism is to allow agents to abstain in this mechanism. Consider a setting where $F_i$ is symmetric around $0$ for each $i$ and adjust the given voting-with-evidence mechanism by allowing for abstention and giving abstentions a weight of 0. Now, an agent with a positive type close enough to 0 strictly prefers to abstain instead of casting a vote in favor, which would give weight $\omega_i^+$. This allows for more information being transmitted to the principal without
adding verification costs and this mechanism can therefore increase the principal's expected utility compared to a voting-with-evidence mechanism. 

\paragraph{Limited commitment}
Following the standard approach in mechanism design, we assume the principal commits to a mechanism. 
There are several ways in which our optimal mechanism uses commitment of the principal and our results would change if the principal could not commit. Most importantly, the principal commits to costly verifications although in equilibrium he will never find an agent lying. Secondly, as explained above, the decision rule is not the first-best for the principal since he distorts the decision by bunching agents and by using net types. This is similar to the use of commitment in standard mechanism design, where principal's often commit to ex-post inefficient outcomes. Thirdly, in our model the principal commits to penalize an agent that is found lying. Note however that it is not necessary to use this third component to commit to unreasonable penalties. Suppose in a voting-with-evidence mechanism agent $i$ deviates and reports $t_i'$ although his true type is $t_i$. This will only be relevant if agent $i$'s report changes the outcome to the more preferred one for him, which implies that agent $i$'s report is decisive. In this case, his report is audited and the penalty for agent $i$ will be to do the opposite of what agent $i$ prefers. This coincides with the decision if agent $i$ would have been truthful and reported $t_i$ in the first place because agent $i$ is decisive. In this sense it is not necessary to use commitment to carry out unreasonable penalties.

The fact that commitment matters is typical for models of costly verification, and contrasts with some models of evidence which show that commitment is not necessary (see, e.g., \citeasnoun{ben-porath17}). One reason for the difference is that, with costly verification, the principal will anticipate the verification costs induced by a given decision rule and deviate from the first-best rule in order to reduce these costs. This effect is not present in models with evidence that have no verification costs.

\newpage
\part*{}
\begin{appendix}
\section{Appendix}

\subsection{Revelation principle}\label{sec_rev_principle}

In this section of the Appendix we show that it is without loss of generality to restrict attention to the class of direct mechanisms as we define them in Section~\ref{sec_model}. Similar versions of the revelation principle have been obtained in \citeasnoun{townsend88} and \citeasnoun{ben-porath14}. We will proceed in two steps. The first step is a revelation principle argument where we establish that any indirect mechanism can be implemented via a direct mechanism. In the second step we show that direct mechanisms can be expressed as a tuple  $(d,a,\ell)$, where $d$ specifies the decision, $a_i$ specifies if agent $i$ is verified, and $\ell_i$ specifies what happens if agent $i$ is revealed to be lying. 

\medskip

\emph{Step 1: It is without loss of generality to  restrict attention to direct mechanisms in which truth-telling is a Bayes-Nash equilibrium.}

Let $(M_1,...,M_I, \tilde{x},\tilde{y})$ be an indirect mechanism, and  $M = \bigtimes_{i \in \mathcal{I}} M_i$, where each $M_i$ denotes the message space for agent $i$, $\tilde{x}: M \times T \times [0,1] \rightarrow \{0,1\}$ is the decision function specifying whether the policy is implemented, and $\tilde{y}: M \times T \times \mathcal{I} \times [0,1] \rightarrow \{0,1\}$ is the verification function specifying whether an agent is verified.\footnote{To describe possibly stochastic mechanisms we introduce a random variable $s$ that is uniformly distributed on $[0,1]$ and only observed by the principal. This random variable is one way to correlate the verification and the decision on the policy.} Fix a Bayes-Nash equilibrium $\sigma$ of the game induced by the indirect mechanism.\footnote{In the game induced by the indirect mechanism, whenever the principal verifies agent $i$ nature draws a type $\tilde{t}_i \in T_i$ as the outcome of the verification. Perfect verification implies that $\tilde{t}_i$ equals the true type of agent $i$ with probability 1. The strategies $m_i \in M_i$ specify an action for each information set where agent $i$ takes an action, even if this information set is never reached with strictly positive probability. In particular, they specify actions for information sets in which the outcome of the verification does not agree with the true type.} 

In the corresponding direct mechanism, let $T_i$ be the message space for agent $i$. Define $x: T \times T \times [0,1] \rightarrow \{0,1\}$ as $x(t',t,s) = \tilde{x}( \sigma(t'),t,s )$ and  $y: T \times T \times \mathcal{I} \times [0,1] \rightarrow \{0,1\}$ as $y(t',t,i,s) = \tilde{y}( \sigma(t'),t,i,s )$. Since $\sigma$ is a Bayes-Nash equilibrium in the original game, truth-telling is a Bayes-Nash equilibrium in the game induced by the direct mechanism. This implies that in both equilibria the same decision is taken and the same agents are verified.

Note that in any feasible direct mechanism the decision whether or not to verify an agent cannot depend on his true type, hence $y(t'_i, t_{-i},t'_i,t_{-i},i,s)=y(t'_i, t_{-i},t,i,s)$. Also, if agent $i$ was not verified, the implementation decision cannot depend on his true type, $x(t,t,s) = x(t,t_i',t_{-i},s)$.

\medskip

\emph{Step 2: Any direct mechanism can be written as a tuple $(d,a,\ell)$, where \\
$d:T\times [0,1] \rightarrow \{0,1\}$, $a_i:T\times [0,1] \rightarrow \{0,1\}$, and $\ell_i:T \times T_i \times [0,1] \rightarrow \{0,1\}$.}

Let
\begin{align*}
  d(t,s) &= x(t,t,s)\\
  a_i(t,s) &= y(t,t,i,s) \text{ and}\\
  \ell_i(t_i',t_{-i},t_i,s) &= x(t_i',t_{-i},t_i,t_{-i},s).
\end{align*}

On the equilibrium path $(d,a,\ell)$ implements the same outcome as $(x,y)$ by definition. Suppose instead agent $i$ of type $t_i$ reports $t_i'$ and all other agents report $t_{-i}$ truthfully. Denoting $t'=(t_i',t_{-i})$, the decision taken in the mechanism $(d,a, \ell)$ if the type profile is $t$ and the report profile is $t'$ is
\begin{align*}
  &[1-a_i(t',s)] d(t',s) + a_i(t',s) \ \ell_i(t_i',t_i,t_{-i},s)\\
  =\ &[1-y(t',t',i,s)] x(t',t',s) + y(t',t',i,s)\ x(t',t,s)\\
  =&\begin{cases}
    x(t',t,s) \text{ if } y(t',t',i,s) = 1 \\
    x(t',t',s) \text{ if } y(t',t',i,s) = 0,
  \end{cases}
\end{align*}

If $y(t',t',i,s) = 1$, the decision is $x(t',t,s)$ under both formulations. Instead, if $y(t',t',i,s) = 0$ then $y(t',t,i,s) = 0$ (since the decision to verify agent $i$ cannot depend on his true type), and hence the decision on the policy must coincide with the case when agent $i$ is verified and reports $t'_i$,  $x(t',t',s) = x(t',t,s)$  We conclude that the decision is the same in both formulations of the mechanism if one agent deviates. Since truth-telling is an equilibrium in the mechanism $(x,y)$, it is therefore an equilibrium in the mechanism $(d,a,\ell)$, which consequently implements the same decision and verification rules.

\subsection{Proof of Theorem \ref{thm_opt}\label{sec_proofmainresult}}
  
In this section of the Appendix  we show that a voting-with-evidence mechanism maximizes the expected utility of the principal. The first step in the  proof of  Theorem \ref{thm_opt}  is to construct a relaxed problem for the principal where the optimization is only over decision rules, compared to maximizing jointly of decision and verification rules in the original problem. The solution to the relaxed problem always yields weakly higher value than the solution to the original optimization problem (Lemma \ref{relax}).  In the second step we show that the solution to the relaxed problem is a voting-with evidence mechanism: first we establish this for finite type spaces (Lemma \ref{prop:opt}) and then extend the result to infinite type spaces (Lemma \ref{prop:opt_inf}). To finish the proof we construct verification rules such that the solution to the relaxed problem is feasible for the original problem and achieves the same objective value. This proves Theorem \ref{thm_opt}.

\renewcommand{\c}{\tilde{c}}
We will show that the problem below is a relaxed version of the principal's  maximization problem as defined in~(\ref{P}):
\begin{align*} \label{r}
  \underset{0\leq d \leq 1} \max  \E_{t}  \Big[\textstyle \sum\limits_id(t)[t_i-\c_i(t_i)] + c_i\Big(\mathbbm{1}_{T_i^+}(t_i) \underset{t_i' \in T_i^+} \inf \E_{t_{-i}}[d(t_i',t_{-i})] - \mathbbm{1}_{T_i^-}(t_i) \underset{t_i' \in T_i^-}\sup \E_{t_{-i}}[d(t_i',t_{-i})]\Big)\Big] \tag{R}
\end{align*}
where $\mathbbm{1}_{T_i^+}(t_i)$ denotes the indicator function for $T_i^+$, $\mathbbm{1}_{T_i^-}(t_i)$ the indicator function for $T_i^-$, and  $\c_i(t_i)=c_i$ if $t_i\in T_i^+$ and  $\c_i(t_i)=-c_i$ if $t_i\in T_i^-$.  

For each mechanism $(d,a)$  let $V_P(d,a)$ denote value of the objective in problem~(\ref{P}), and for each decision rule $d$ let $V_R(d)$ denote the objective value  in problem \eqref{r}.

\begin{lemma}\label{relax}
  For any Bayesian incentive compatible mechanism $(d,a)$, $V_P(d,a) \le V_R(d)$.
\end{lemma}

\begin{proof}
\small
  \begin{align} 
  &V_{P}(d,a) =\E_{t}  \left[\sum_i d(t)[t_i-\c_i(t_i)] + c_i \mathbbm{1}_{T_i^+}(t_i) [d(t)-a_i(t)] - c_i \mathbbm{1}_{T_i^-}(t_i) [d(t)+a_i(t)] \right]\nonumber \\
  \le &\E_{t}  \left[\sum_i d(t)[t_i-\c_i(t_i)] + c_i \mathbbm{1}_{T_i^+}(t_i) [d(t) (1- a_i(t))] - c_i \mathbbm{1}_{T_i^-}(t_i) [d(t) (1- a_i(t))+a_i(t)] \right] \label{first_ineq}\\
  \le &\E_{t}  \left[\sum_i d(t)[t_i-\c_i(t_i)] + c_i \mathbbm{1}_{T_i^+}(t_i) \inf_{t_i' \in T_i^+} \E_{t_{-i}}[d(t_i',t_{-i})] - c_i \mathbbm{1}_{T_i^-}(t_i) \sup_{t_i' \in T_i^-} \E_{t_{-i}}[d(t_i',t_{-i})] \right] \label{second_ineq}\\
  = &V_{R}(d)\nonumber.
\end{align}
\normalsize
The first inequality holds because $-a_i(t) \le -d(t)a_i(t)$ and $d(t)a_i(t)\ge 0$.
The second inequality follows from the fact that $(d,a)$ is BIC.
\end{proof}

The significance of the relaxed problem lies in the fact that for any optimal solution $d$ to problem \eqref{r}, we can construct verification rules $a$ such that $(d,a)$ is feasible and $V_P(d,a)=V_R(d)$. This implies that $d$ is part of an optimal solution to problem \eqref{P}.

We now describe an optimal solution to the relaxed problem for finite type spaces.

\begin{lemma}\label{prop:opt}
Suppose that the type space $T$ is finite. Problem \eqref{r} is solved by a voting-with-evidence mechanism.
\end{lemma}

\begin{proof}
Let $d^*$ denote an optimal solution to~\eqref{r}, let $\varphi_i^+\equiv \inf_{t_i'\in T_i^+}  \E_{t_{-i}}[d^*(t'_i,t_{-i})] $ and $\varphi_i^-\equiv \sup_{t_i'\in T_i^-} \E_{t_{-i}}[d^*(t'_i,t_{-i})]$, and observe that $\varphi_i^-\le \varphi_i^+$.

Consider the following auxiliary maximization problem:
  \begin{align*} \label{aux}
    \underset{0 \leq d\leq 1 }\max \ &\E_{t}  \big[\textstyle\sum_id(t)[t_i-\c_i(t_i)] \big]\tag{Aux}\\
    &   \hspace{1cm}  \text{s.t.  for all $i\in\I$:  }    \\
    &   \hspace{1cm}  \text{$ \E_{t_{-i}} [d(t)]\geq \varphi_i^+ $ for all $t_i\in T_i^+$, and }\\
    &   \hspace{1cm} \text{$ \E_{t_{-i}} [d(t)]\leq \varphi_i^-  $ for all $ t_i\in T_i^-$,}
  \end{align*}
Clearly, $d^*$ also solves the auxiliary problem. The Karush-Kuhn-Tucker theorem \cite{arrow61,luenberger69} implies that there exist Lagrange multipliers $\lambda^*_i(t_i)$, such that $\lambda^*_i(t_i)\ge 0$ for $t_i\in T_i^+$ and $\lambda^*_i(t_i)\le 0$ for $t_i \in T_i^-$ and such that $d^\ast$ maximizes

\begin{align*}
  \mathcal{L}(d,\lambda^*) &= \E_t\Big[\sum_i d(t)(t_i-\c_i(t_i))\Big] + \sum_i \sum_{t_i \in T_i} \Big( \lambda_i^*(t_i) \big( \E_{t_{-i}}[d(t_i,t_{-i})] -\varphi_i(t_i) \big) \Big)\\
  &=\sum_{t \in T}{ d(t) \sum_i \Big(t_i - \c_i(t_i) + \frac{\lambda_i^*(t_{i})}{ f_i(t_i) } \Big) f(t) } + constant,
\end{align*}
where 
\begin{minipage}{7cm}
  \[
  \varphi_i(t_i):=
  \left\{
  \begin{array}{ccc}
  \varphi_i^+ &\text{ if }& t_i\in T_i^+\\
  \varphi_i^- & \text{ if } &t_i\in T_i^-.
  \end{array}
  \right.
  \]
  \end{minipage}
\vspace{5mm}

Setting $h^*_i(t_i) := t_i -c_i(t_i) + \frac{\lambda_i^*(t_{i})}{f_i(t_i)}$ and ignoring the constant in the Lagrangian, we observe that $d^*$ maximizes the function
\begin{equation}\label{lagrangian}
 g(d,h^*) = \sum_{t \in T} \sum_i d(t)f(t)  h^*_i(t_i).
\end{equation}

Let
\begin{eqnarray}
  \alpha_i^+ =& \inf_{t_i \in T_i^+}\{ t_i | \E_{t_{-i}}[d^*(t_i,t_{-i})]  > \varphi_i^+ \} - c_i \\
  \alpha_i^- =& \sup_{t_i \in T_i^-}\{ t_i | \E_{t_{-i}}[d^*(t_i,t_{-i})] < \varphi_i^- \}+c_i
  \end{eqnarray}
and define
  \begin{align*}
    \bar{h}_i(t_i) := \begin{cases}
      \frac{1}{\mu_i(A_i^+)} \sum_{t_i\in A_i^+} f_i(t_i) h_i^*(t_i) \hspace{.3cm} &\text{ if } t_i \in T_i^+ \text{ and }  t_i \le \alpha_i^{+} +c_i \\
     \frac{1}{\mu_i(A_i^-)} \sum_{t_i\in A_i^-} f_i(t_i) h_i^*(t_i) \hspace{.3cm} &\text{ if } t_i \in T_i^- \text{ and } t_i \ge \alpha_i^{-}-c_i \\
      t_i-\c_i(t_i) &\text{ otherwise},
    \end{cases}
  \end{align*}
where  $A_i^+= \{t_i\in T_i^+| t_i< \alpha^+_i +c_i \}$, $A_i^-= \{t_i\in T_i^-| t_i > \alpha^-_i -c_i \}$, and $\mu_i(A)$ denotes the measure induced by $F_i$. Let  $A^c_i= T_i\setminus(A_i^+\cup A_i^-)$ and $A_i=A_i^+\cup A_i^-$.

\begin{claim}\label{claim_omega}
$d^*$ also maximizes $g(d,\bar{h}) = \sum_{t \in T} \sum_i d(t)f(t) \bar{h}_i(t_i)$.
\end{claim}

\noindent
\textbf{Step 1:} $\lambda^*(t_i)=0$ for $t_i\in A_i^c$.

Complementary slackness implies $\lambda_i^*(\alpha_i^+ +c_i) =0$. Moreover, for every $t_i\in T_i^+$ such that $t_i > \alpha_i^+ +c_i$, we get $t_i-c_i+\frac{\lambda^*_i(t_i)}{f_i(t_i)}\ge \alpha_i^+$ and hence for every optimal solution to the Lagrangian $d$ that $\E_{t_{-i}}[d(t_i,t_{-i})] \ge \E_{t_{-i}}[d(\alpha_i^++c_i,t_{-i})] > \varphi_i^+$.
 This implies that for $t_i\in T_i^+\cap A_i^c$, $\lambda_i^*(t_i) =0$  by complementary slackness. Analogous arguments for $t_i \in T_i^-\cap A_i^c$ apply. Thus, $\lambda^*(t_i)=0$ for $t_i\in A_i^c$.

\noindent 
\textbf{Step 2:}  $g(d^*,h^*) = g(d^*,\bar{h})$.

First, observe that $h^*_i(t_i)=\bar{h}_i(t_i)$ for $t_i \in A_i^c$, $\varphi_i^+=\Ei [d^*(t_i,t_{-i})]$ for $t_i\in A_i^+$, and  $\varphi_i^-=\Ei [d^*(t_i,t_{-i})]$ for $t_i\in A_i^-$. This implies 

\[
\begin{split}
g(d^*,h^*)&
=\sum_i \Big[\sum_{t_i\in A_i}h^*_i(t_i)f_i(t_i)\Ei [d^*(t)] + \sum_{t_i\in A_i^c}h^*_i(t_i)f_i(t_i)\Ei [d^*(t)]\Big] \\
&= \sum_i \Big[\sum_{t_i\in A_i^+}h^*_i(t_i)f_i(t_i) \varphi_i^++ \sum_{t_i\in A_i^-}h^*_i(t_i)f_i(t_i) \varphi_i^-+ \sum_{t_i\in A_i^c}\bar{h}_i(t_i)f_i(t_i)\Ei [d^*(t)]\Big] \\
&= \sum_i \Big[\sum_{t_i\in A_i^+} \bar{h}_i(t_i)f_i(t_i) \varphi_i^++ \sum_{t_i\in A_i^-} \bar{h}_i(t_i)f_i(t_i)\varphi_i^-+ \sum_{t_i\in A_i^c}\bar{h}_i(t_i)f_i(t_i)\Ei [d^*(t)]\Big] \\
&= \sum_i \Big[\sum_{t_i\in A_i}\bar{h}_i(t_i)f_i(t_i)\Ei [d^*(t)] + \sum_{t_i\in A_i^c}\bar{h}_i(t_i)f_i(t_i)\Ei [d^*(t)]\Big] = g(d^*,\bar{h}).
\end{split} 
 \]
\vspace{-.4cm}

\noindent
\textbf{Step 3:} $g(d^*,\bar{h})=g(d^*,h^*)=\max_{0\le d\le 1} g(d,h^*) \ge \max_{0\le d\le 1} g(d,\bar{h})$.

The first equality follows from Step 2 and the second holds because $d^*$ maximizes $g(d,h^*)$ by construction. 

Let  $h_i:T_i\rightarrow \mathbb{R}$ be any real-valued function, and for each such  function $h_i$ define    $H_i(t_i):= h_i(t_i)f_i(t_i)$ and denote by $H_i\equiv (H_i(t_i))_{t_i\in T_i}$. Fix an agent  $i\in \I$, and define a function $\Psi:\mathbb{R}^{|T_i|}\rightarrow\mathbb{R}$, as $\Psi(H_i):= \max_{0\le d\le 1} \sum_{t\in T} d(t)\big[   f_{-i}(t_{-i}) H_i(t_i) + \sum_{j\in \I_{-i}} f(t)h^*_j(t_j)\big ]  $.  The function $\Psi$ is convex, since it is a maximum over linear functions. It is also  symmetric, since permuting the vector $H_i$ does not change  the value of $\Psi$. Thus, $\Psi $ is Schur-convex. By construction, $H^*_i$ (defined as $H^*_i(t_i)=h^*_i(t_i)f_i(t_i)$)  majorizes $\bar{H}_i$ (defined as $\bar{H}_i(t_i)=\bar{h}_i(t_i)f_i(t_i)$).
Therefore we obtain that,
\[
 \Psi(H_i^*)\geq \Psi(\bar{H}_i) 
\]
We have now shown that if we replace $h_i^*$  for agent $i$ with its average $\bar{h}_i$ we have that $d^*$ remains the maximizer of   $ \max_{0\le d\le 1 }g(d,h^*_{\I_{-i}}h_i)$. By repeating this argument agent by agent we can conclude that,
\[ \max_{0\le d\le 1} g(d,h^*) =\max_{0\le d\le 1} \sum_{t\in T} \sum_{i\in\I} d(t)f_{-i}(t_{-i})H^*_i(t_i) \ge \max_{0\le d\le 1} \sum_{t\in T} \sum_{i\in\I} d(t)f_{-i}(t_{-i})\bar{H}_i(t_i) = \max_{0\le d\le 1} g(d,\bar{h}). \]
This proves the Claim~\ref{claim_omega}.

\bigskip

Hence, every solution to the Lagrangian can be described as follows:
  \begin{align*}
       d(t) = \begin{cases}
              1 \hspace{.3cm} &\text{ if } \sum{w_i(t_i)} > 0 \\
              0 &\text{ if } \sum{w_i(t_i)} < 0,
    \end{cases}
  \end{align*}
where
  \begin{align}\label{eq_mech}
    w_i(t_i) = \begin{cases}
     \omega_i^{+} \hspace{.3cm} &\text{ if } t_i \in T_i^+ \text{ and }  t_i \le \alpha_i^{+} +c_i \\
      \omega_i^{-} &\text{ if } t_i \in T_i^- \text{ and } t_i \ge \alpha_i^{-}-c_i \\
      t_i - c_i(t_i) &\text{ otherwise}
    \end{cases}
  \end{align}
for constants $\{\omega_i^+,\omega_i^-\}_{i\in I}$. Since $d^*$ maximizes the Lagrangian by assumption, we conclude that it takes this form.

Note that $\omega_i^+ \ge \sup_{t_i\in A_i^+}\{t_i-c_i\}$ since $\lambda_i^*(t_i)\ge 0$ for $t_i \in A_i^+$. Also, $\omega_i^+\le \alpha_i^+$, since otherwise we would get, for $t_i\in A_i^+$, $\Ei[d^*(t_i,t_{-i})]\ge\Ei[d^*(\alpha_i^+-c_i,t_{-i})]>\varphi_i^+$, contradicting the definition of $A_i^+$. Analogous arguments imply $\inf_{t_i\in A_i^-}\{t_i+c_i\}\le \omega^- \le \alpha_i^-$.  This implies that  we can replace $\alpha_i^+$ ($\alpha_i^-$ ) with $\omega_i^+$ ($\omega_i^-$) in the definition of the weight function $w_i$ in \eqref{eq_mech} above without changing the outcome of the  mechanism in any way.
\end{proof}

As the next step in the proof we show that voting-with-evidence mechanisms are also optimal for infinite type space.

\begin{lemma}\label{prop:opt_inf}
 Suppose that $T$ is an infinite type space.  Problem \eqref{r} is solved by a voting-with-evidence mechanism.
\end{lemma}
\begin{proof}
Let $F_i^+$ and $F_i^-$ denote the conditional distributions induced by $F_i$ on $T_i^+$ and $T_i^-$, respectively.
We first construct a discrete approximation of the type space: For $i \in \I$, $n \ge 1$, $l_i = 1,\dots,2^{n+1}$, let 
\begin{align*}
	S_i(n,l_i) := \begin{cases} \{t_i \in T_i^+ | \frac{l_i-1}{2^n} \le F_i^+(t_i) < \frac{l_i}{2^n} \} &\text{ for } l_i \le 2^n \\
	\{t_i \in T_i^- | \frac{l_i-2^n-1}{2^n} \le F_i^-(t_i) < \frac{l_i-2^n}{2^n} \} &\text{ for } l_i > 2^n,
	\end{cases}
\end{align*}
 which form partitions of $T_i^+$ and $T_i^-$, and denote by $\mathcal{F}_i^n$ the set consisting of all possible  unions of the $S_i(n,l_i)$. Let $l = (l_1, ..., l_n)$ and $S(n,l) = \prod_{i \in \I} S_i(n,l_i)$, which defines a partition of $T$, and denote by $\mathcal{F}^n$ the induced $\sigma$-algebra. 

Let $(R^n)$ denote the relaxed problem with the additional restriction that $d$ is measurable with respect to $\mathcal{F}^n$. Then the constraint set has non-empty interior and an optimal solution to $(R^n)$ exists.
Define $\tilde{t}_i(t_i)~:=~\frac{1}{\mu_i(S_i(n,l_i))} \int_{S_i(n,l_i)} s dF_i $ for $t_i \in S_i(n,l_i)$, where $\mu_i$ denotes the measure induced by $F_i$. The arguments for finite type spaces imply that the following rule is an optimal solution to $(R^n)$ for some $\omega_{i}^{+,n}, \omega_i^{-,n}$:
\begin{align*}
	r^n_i(t_i) = \begin{cases}
		\omega_i^{+,n} - c_i \hspace{.3cm} &\text{ if } t_i \in T_i^+  \text{ and } \tilde{t}_i(t_i) \le \omega_i^{+,n} \\
		\omega_i^{-,n} +c_i &\text{ if } t_i \in T_i^- \text{ and } \tilde{t}_i(t_i)\ge \omega_i^{-,n}\\
		\tilde{t}_i(t_i) - c_i(t_i) &\text{ otherwise}
	\end{cases}
\end{align*}
\begin{align*}
	d^n(t) = \begin{cases}
		1 \hspace{.3cm} &\text{ if } \sum{r_i^n(t_i)} > 0 \\
		0 &\text{ if } \sum{r_i^n(t_i)} < 0.	
	\end{cases}
\end{align*}

Let $\omega_i^+ := \underset{n\rightarrow \infty}\lim \, \omega_i^{+n}$ and $\omega_i^- := \underset{n\rightarrow \infty}\lim \,\omega_i^{-,n}$ (by potentially choosing a convergent subsequence). 
Define 
\begin{align*}
	r_i(t_i) = \begin{cases}
		\omega_i^+ -c_i \hspace{.3cm} &\text{ if } t_i \in T_i^+  \text{ and } \tilde{t}_i(t_i) \le \omega_i^{+,n} \\
		\omega_i^- + c_i &\text{ if }t_i \in T_i^- \text{ and } \tilde{t}_i(t_i)\ge \omega_i^{-,n}\\
		t_i - \c_i(t_i) &\text{ otherwise}
	\end{cases}
\end{align*}
\begin{align*}
	d(t) = \begin{cases}
		1 \hspace{.3cm} &\text{ if } \sum{r_i(t_i)} > 0 \\
		0 &\text{ if } \sum{r_i(t_i)} < 0.
	\end{cases}
\end{align*}

Then, for all $i$ and $t_i$, $\E_{t_-i}[d^n(t_i,t_{-i})] = \Prob[ \sum_{j \neq i} r_j^n(t_j) \ge - r^n_i(t_i) ]$ converges pointwise almost everywhere to $ \E_{t_{-i}}[d(t_i,t_{-i})] $. This implies that the marginals converge in $L^1$-norm and hence the objective value of $d^n$ converges to the objective value of $d$. 
This implies that $d$ is an optimal solution to \eqref{r}, since if there was a solution achieving a strictly higher objective value, there would exist $\mathcal{F}^n$-measurable solutions achieving a strictly higher objective value for all $n$ large enough. Therefore, a voting-with-evidence mechanism solves problem \eqref{r}.
\end{proof}

Now we have all the parts required to establish our main result Theorem~\ref{thm_opt} that voting-with-evidence mechanisms are optimal. 
\begin{proof}[Proof of Theorem \ref{thm_opt}]\label{pf_thm_opt} 
  Denote by $d^*$ the solution to problem \eqref{r}. We first construct a verification rule $a^*$ such that $(d^*,a^*)$ is Bayesian incentive compatible and then argue that $V_{P}(d^*,a^*) = V_{R}(d^*)$. Given that $V_{P}(d,a) \le V_{R}(d)$ holds for any incentive compatible mechanism, this implies that $(d^*,a^*)$ solves \eqref{P}.

  Let $a^*$ be such that agent $i$ is verified whenever he is decisive. Then $a^*_i(t) = a^*_i(t) d^*(t)$ for all $t_i \in T_i^+$ (if $d^*(t)=0$ then type $t_i \in T_i^+$ is not decisive), and $d^*(t) = d^*(t) [1-a^*_i(t)] $ for all $t_i \in T_i^-$ (if $a^*_i(t)=1$ then $d^*(t)=0$). Hence, inequality \eqref{first_ineq} holds as an equality for $(d^*,a^*)$. 

  Note that in mechanism $(d^*,a^*)$, all incentive constraints are binding and therefore inequality \eqref{second_ineq} holds as an equality as well. We therefore conclude $V_{P}(d^*,a^*) = V_{R}(d^*)$.
\end{proof}

\begin{proof}[Proof of Proposition \ref{prop:opt_weights}]
Without loss of generality, suppose $\omega_1^+\le-\omega_2^-$ and consider changing $\omega_2^-$ (the other cases are analogous). This matters only if agent $2$ has a negative type and 1 has a positive type. We consider two cases: (a) a change to $\omega_2^{-'}$ such that $\omega_1^+\le-\omega_2^{-'}$; (b) a change such that $\omega_1^+>-\omega_2^{-'}$.

Case (a): 

Using weight $\omega_2^{-'}$ such that $-\omega_1^+\ge\omega_2^{-'}>\omega_2^-$ instead of $\omega_2^-$ matters only if agent 1's type satisfies $\omega_2^{-}\le -t_1+c_1 \le \omega_2^{-'}$. Conditional on such a type, the expected utility of the principal from using weight $\omega_2^-$ is 0. On the other hand, using weight $\omega_2^{-'}$ gives conditional expected utility of
\begin{align*}
\int_{-\infty}^0 (t_1+t_2-c_1)\1_{t_1-c_1+t_2+c_2\ge 0}-c_2 \1_{t_1-c_1+t_2+c_2<0} dF_2.
\end{align*}
The definition of $\omega_2^-$ implies
\begin{align*}
\int_{-\infty}^0 t_2 dF_2 &= \int_{-\infty}^0 \min\{\omega_2^-, t_2+c_2\} dF_2 \\
& \le \int_{-\infty}^0 \min\{-t_1+c_1, t_2+c_2\} dF_2 \\
&=\int_{-\infty}^0 (-t_1+c_1)\1_{t_1-c_1+t_2+c_2\ge 0}+(t_2+c_2) \1_{t_1-c_1+t_2+c_2<0} dF_2.
\end{align*}
Subtracting $\int_{-\infty}^0 t_2 dF_2$ from both sides and multiplying by $-1$, this implies
\begin{align*}
0\ge\int_{-\infty}^0 (t_1+t_2-c_1)\1_{t_1-c_1+t_2+c_2\ge 0}-c_2 \1_{t_1-c_1+t_2+c_2<0} dF_2,
\end{align*}
and hence the principal is better off using weight $\omega_2^-$. Similar arguments also show that the principal is worse off using a cutoff $\omega_2^{-'}<\omega_2^-$.

Case (b): We can think of this case in two steps. First, a change such that $\hat{\omega}_2^-=-\omega_1^+$. As shown in Case (a), this reduces the principal's welfare. Second, a further change to $\omega_2^{-'}$, which only changes the decision if both agents cast a vote. The effect of this second change non-positive if and only if
\begin{align}
0\ge &\int_0^{\omega_1^+ + c_1} \int_{-\omega_1^- - c_2}^0 t_1+t_2 dF_2\ dF_1\nonumber \\
&+ c_1 [1-F_1(\omega_1^+ +c_1)][F_2(0)-F_2(-\omega_1^- -c_2)] - c_2 [F_1(\omega_1^+ + c_1)-F_1(0)]F_2(-\omega_1^- -c_2). \nonumber
\end{align}
This is equivalent to
\begin{align}
0\ge & \big\{\E[t_1|0\le t_1\le \omega_1^+ + c_1] +\E[t_2|0\ge t_2\ge -\omega_1^+ - c_2] \big\}[F_1(\omega_1^+ + c_1)-F_1(0)][F_2(0)-F_2(-\omega_1^+ - c_2)]\nonumber \\
&+ c_1 [1-F_1(\omega_1^+ +c_1)][F_2(0)-F_2(-\omega_1^- -c_2)] - c_2 [F_1(\omega_1^+ + c_1)-F_1(0)]F_2(-\omega_1^- -c_2) \nonumber,
\end{align}
or to
\begin{align}
0\ge & \E[t_1|0\le t_1\le \omega_1^+ + c_1] +\E[t_2|0\ge t_2\ge -\omega_1^+ - c_2] \nonumber \\
 &+ c_1 \frac{1-F_1(0)}{F_1(\omega_1^+ + c_1)-F_1(0)} -c_1  - c_2 \frac{F_2(0)}{F_2(0)-F_2(-\omega_1^+ - c_2)}+c_2 \label{eq:n2}
\end{align}
However, the definition of $\omega_1^+$ implies
\begin{align}
\int_0^{\infty} t_1 dF_1 &= \int_{\omega_1^+ + c_1}^{\infty} t_1 - c_1 d F_1 + [F(\omega_1^+ + c_1)-F(0)]\omega_1^+\nonumber\\
\Leftrightarrow  \omega_1^+ &= \E[t_1|0\le t_1\le \omega_1^+ + c_1] -c_1+ c_1 \frac{1-F_1(0)}{F_1(\omega_1^+  + c_1)-F_1(0)}.\label{eq:n2a}
\end{align}
Similarly, the definition of $\omega_2^-$ and the fact that $\omega_2^-\le -\omega_1^+$ imply
\begin{align*}
\E[t_2|t_2< 0] &= \E[\min\{\omega_2^- -c_2,t_2\} +c_2| t_2<0]\\
& \le \E[\min\{-\omega_1^+ -c_2,t_2\} +c_2| t_2<0].
\end{align*}
Rearranging this inequality yields
\begin{align}\label{eq:n2b}
\E[t_2|-\omega_1^+- c_2\le t_2< 0] -c_2 \frac{F_2(0)}{F_2(0)-F_2(\omega_1^+ -c_2)}\le -\omega_1^+ -c_2.
\end{align}
Plugging \eqref{eq:n2a} and \eqref{eq:n2b} in \eqref{eq:n2}, we see that \eqref{eq:n2} holds. We conclude that the principal is better off using weight $\omega_2^-$.
\end{proof}

\subsection{Omitted proofs from Section \ref{sec_bicepic}} \label{proof_bicepic}

\begin{proof}[Proof of Theorem \ref{thm:commute}] \ 
\newline		
The proof applies Theorem 6 in \cite{gutmann91} to a discrete approximation of $A$ and by taking limits we establish Theorem~\ref{thm:commute}. 
		
		Let $S_i(n,l_i)$ denote the interval,
		\[
		S_i(n,l_i) := [F^{-1}_i((l_i-1)2^{-n}), F^{-1}_i(l_i 2^{-n})),\hspace{1 cm} i \in \I,\, n \ge 1 \text{ and } l_i = 1,...,2^n.
		\]
For a given $n$ the function $S_i(n,\cdot)$ form a partition of $A_i$ such that each partition element $S_i(n,k)$ has the same likelihood. Let $\mathcal{F}_i^n$ denote  the set consisting of all possible  unions of the $S_i(n,l_i)$. Note further that $\mathcal{F}_i^n\subset \mathcal{F}_i^{n+1}$. Let $l =(l_1,\dots,l_I)$ and $S(n,l):= \prod_{i\in\I}S_i(n,l_i)$. Thus, for a given $n$ the function $S(n,\cdot)$ defines a partition of $A$ such that each partition element $S(n,l)$  has the same likelihood.

Define the following averaged function,
\[
g(n,l) := 2^{In} \int_{S(n,l)} g(t) dF. 
\]
The function $g(n,l)$ is an $I$-dimensional tensor. Now consider the marginals of $g(n,l)$ with respect to $l_{-i}$, i.e., $\E_{l_{-i}}[g(n,l_i,l_{-i})]$, each such marginal in dimension $i$ is nondecreasing in $l_i$. By Theorem 6 in \cite{gutmann91} there exists another tensor $g'(n,l)$ with the same marginals as $g(n,l)$ such that $g'(n,l)$ is nondecreasing in $l$. Now define $g_n': T \rightarrow [0,1]$ by letting $g_n'(t) := g'(n,l)$ for all $t \in S(n,l)$.

		Note that $g_n'$ is nondecreasing in each coordinate and hence satisfies 
		\begin{align}
			 \int \essinf_{t_i \in B} g_n'(t_i,t_{-i}) dF_{-i} &= \essinf_{t_i \in B} \int g_n'(t_i,t_{-i}) dF_{-i} \label{eq:inf2}\\
			 \int \esssup_{t_i \in B} g_n'(t_i,t_{-i}) dF_{-i} &= \esssup_{t_i \in B} \int g_n'(t_i,t_{-i}) dF_{-i}.\label{eq:sup2}
		\end{align}
		Moreover, 
		\begin{align} \label{integrate_to_zero}
      \int_{S_i(n,l_i)} \int_{A_{-i}} g(t_i,t_{-i}) dF_{-i} \ dF_i  = \int_{S_i(n,l_i)} \int_{A_{-i}} g_n'(t_i,t_{-i}) dF_{-i} \ dF_i,
    \end{align}
		and hence $g(t)-g_n'(t)$ integrates to zero over sets of the form $S_i(n,l_i) \times A_{-i}$ for every $S_i(n,l_i) \in \mathcal{F}^n_i$.

		\bigskip

		Draw a weak$^*$-convergent subsequence from  the sequence $\{g'_n\}$ (which is possible by Alaoglu's theorem) and denote its limit by $\hat{g}$. This function  rule satisfies $0\le \hat{g} \le 1$ and its marginals  are equal almost everywhere to the marginals of $g$ because of \eqref{integrate_to_zero}. 

		Since $g_n' \rightarrow^* \hat{g}$, we get \\
		$\essinf_{t_i \in B} g_n'(t_i, t_{-i}) \rightarrow \essinf_{t_i \in B} \hat{g}(t_i, t_{-i})$ for almost every $t_{-i}$. Moreover, \\
		$\essinf_{t_i \in B} \int_{A_{-i}}g_n'(t_i, t_{-i}) dF_{-i} \rightarrow \essinf_{t_i \in B} \int_{A_{-i}}\hat{g}(t_i, t_{-i}) dF_{-i}$. 
    Note further that,\\  $\E_{t_{-i}} [\inf_{t_i \in T_i^+} \hat{g}(t_i,t_{-i})] \le \inf_{t_i \in T_i^+} \E_{t_{-i}}[\hat{g}(t_i,t_{-i})]$  always holds. 
    By way of contradiction suppose now  that for some $i$,
		\[ \int \essinf_{t_i \in B} \hat{g}(t_i,t_{-i}) dF_{-i} < \essinf_{t_i \in B} \int \hat{g}(t_i,t_{-i}) dF_{-i}.\footnote{If the inequality only holds for the infimum but not for the essential infimum, we can adjust $\hat{g}$ on a set of measure zero such that our claim holds.} \]
		This implies 
		\[ \int \essinf_{t_i \in B} g_n'(t_i,t_{-i}) dF_{-i} < \essinf_{t_i \in B} \int g_n'(t_i,t_{-i}) dF_{-i}\]
		for $n$ large enough, contradicting \eqref{eq:inf2} and thereby proving the first equality in the theorem. Analogous arguments apply for the second equality in the theorem, thus establishing our claim.
\end{proof}

		\bigskip

  \begin{proof}[Proof of Theorem \ref{prop:bic-epic}] \ 
  \newline
It follows from Theorem \ref{thm:commute} that there exists a decision rule $\hat{d}: T \times [0,1] \rightarrow \{0,1\}$ that induces the same marginals almost everywhere and for which 
  \begin{align*}
    \inf_{t_i \in T_i^+}\E_{t_{-i},s}[\hat{d}(t_i,t_{-i},s)] &= \E_{t_{-i}}[\inf_{t_i \in T_i^+}\E_{s}\hat{d}(t_i,t_{-i},s)] \text{ and }\\
    \sup_{t_i \in T_i^-}\E_{t_{-i},s}[\hat{d}(t_i,t_{-i},s)] &= \E_{t_{-i}}[\sup_{t_i \in T_i^-}\E_{s} \hat{d}(t_i,t_{-i},s)].
  \end{align*}

We now construct a verification rule $\hat{a}$ such that the mechanism $(\hat{d},\hat{a})$ satisfies the claim.
By setting 
		\begin{align*}
		 	\hat{a}_i(t,s) := \begin{cases}
		 	\frac{1}{\Prob_s(\hat{d}(t,s)=1)}  \left( \E_{s'}[\hat{d}(t,s')] - \inf_{t_i' \in T_i^+} \E_{s'}[\hat{d}(t_i',\ti,s')] \right) \hspace{.3cm} &\text{ if $\hat{d}(t,s) =1$} \\
		 	\frac{1}{\Prob_s(\hat{d}(t,s)=0)}  \left( \sup_{t_i' \in T_i^-} \E_{s'}[\hat{d}(t_i',\ti,s')] - \E_{s'}[\hat{d}(t,s')] \right) &\text{ if $\hat{d}(t,s) =0$,}
		 	\end{cases}
		 \end{align*}  
		the mechanism $(\hat{d},\hat{a})$ satisfies \eqref{eq:epic1} as an equality for all $t_i$, $t_{-i}$:
		\begin{align*}
		 	&\E_s[\hat{d}(t,s) (1-\hat{a}_i(t,s))] =  \int\limits_{s:\hat{d}(t,s)=1} 1 - \frac{1}{\Prob_s(\hat{d}(t,s)=1)} \ \left[ \E_{s'}[\hat{d}(t,s')] - \inf_{t_i' \in T_i^+} \E_{s'}[\hat{d}(t_i',\ti,s')] \right] ds \\		 	
		 	& = \int\limits_{s:\hat{d}(t,s)=1} 1 - \frac{1}{\Prob_s(\hat{d}(t,s)=1)} \ \left[ \int\limits_{s':\hat{d}(t,s')=1}\Prob_{s'}(\hat{d}(t,s')=1) ds' - \inf_{t_i' \in T_i^+} \E_{s'}[\hat{d}(t_i',\ti,s')] \right] ds\\
		 	 &= \int\limits_{s:\hat{d}(t,s)=1} \frac{1}{\Prob_s(\hat{d}(t,s)=1)} \left[\inf_{t_i' \in T_i^+} \E_{s'}[\hat{d}(t_i',\ti,s')] \right]ds \\
		 	 &= \inf_{t_i' \in T_i^+} \E_{s}[\hat{d}(t_i',t_{-i},s)].
		 \end{align*} 
		Similarly, the mechanism satisfies \eqref{eq:epic2} as an equality and hence it is EPIC. 

		Moreover, 
		\begin{align*}
			\E_{t_{-i},s}[\hat{a}_i(t,s)] &= \E_{t_{-i},s}\Big[\hat{a}_i(t,s) + \hat{d}(t,s) [1-\hat{a}_i(t,s)] -\hat{d}(t,s) [1-\hat{a}_i(t,s)] \Big] \\
			 &= \E_{t_{-i}}\Bigg[\sup_{t_i'\in T_i^-} \E_{s}\hat{d}(t_i',\ti,s) - \inf_{t_i' \in T_i^+} \E_{s}\hat{d}(t_i',t_{-i},s)\Bigg] \\
			  & = \sup_{t_i' \in T_i^-} \E_{t_{-i},s}[d(t_i',\ti,s)] - \inf_{t_i'\in T_i^+} \E_{t_{-i},s}[d(t_i',t_{-i},s)] \\
			&\le \E_{\ti,s}[a_i(t,s)],
		\end{align*}
		where the second equality follows from the fact that \eqref{eq:epic1} and \eqref{eq:epic2} are binding, the third equality follows from Step 1 and the fact that $d$ and $\hat{d}$ induce the same marginals, and the inequality follows from the fact that $(d,a)$ is BIC. Hence, by potentially adding additional verifications one obtains an EPIC mechanism that induces the same interim decision and verification probabilities.
	\end{proof}

\subsection{Proof of Theorem \ref{thm:imp}} \label{proof_imp}
Consider the relaxed problem
\begin{align*}
  \underset{0\leq d \leq 1} \max  \E_{t}  \Big[\textstyle \sum\limits_id(t)[t_i-\frac{ \tilde{c}_i(t_i)}{p}] + \frac{c_i}{p}\Big(\mathbbm{1}_{T_i^+}(t_i) \underset{t_i' \in T_i^+} \inf \E_{t_{-i}}[d(t_i',t_{-i})] &- \mathbbm{1}_{T_i^-}(t_i) \underset{t_i' \in T_i^-}\sup \E_{t_{-i}}[d(t_i',t_{-i})]\Big)\Big]\\
  \text{s.t. } \eqref{eq:bic_imp1} \text{ and } \eqref{eq:bic_imp2} \tag{$\tilde{R}$}\label{relax_imp}
\end{align*}
For any mechanism $(d,a)$, let $V_{\tilde{P}}(d,a)$ denote the expected utility of the principal given mechanism $(d,a)$ and let $V_{\tilde{R}}(d)$ denote the value achieved by the decision rule $d$ in the relaxed problem. 
\begin{lemma}
For any mechanism $(d,a)$ that is Bayesian incentive compatible in the imperfect verification setting, $V_{\tilde{P}}(d,a)\le V_{\tilde{R}}(d)$.
\end{lemma}

\begin{proof}
Note that Lemma \ref{lemma_bic_impp} implies that 
\begin{align}
\forall t_i \in T_i^+: \ \Eis[a_i(t_i,\ti,s)d(t_i,\ti,s)]\ge \frac{1}{p} \left[\Eis[d(t_i,\ti)] - \inf_{t_i \in T_i^+} \Eis[d(t_i',\ti,s)]\right] \text{ and}\\
\forall t_i \in T_i^-: \ \Eis[a_i(t_i,\ti,s)[1-d(t_i,\ti,s)]] \ge \frac{1}{p}\left[ \sup_{t_i \in T_i^+} \Eis[d(t_i',\ti,s)]-\Eis[d(t_i,\ti)]\right]
\end{align}
Hence,
\small
  \begin{align} 
  V_{\tilde{P}}(d,a) &=\E_{t}  \left[\sum_id(t)t_i- a_i(t)c_i \right]\nonumber \\
  &\le \E_{t}  \left[\sum_id(t)t_i- \mathbbm{1}_{T_i^+}(t_i) d(t)a_i(t)c_i - \mathbbm{1}_{T_i^-}(t_i)  [1-d(t)]a_i(t)c_i \right]\label{eq:ineq_imp1} \\
    &\le \E_{t_i}  \left[\sum_i\Ei [d(t)]t_i- \mathbbm{1}_{T_i^+}(t_i) \frac{1}{p} \left[\Eis[d(t_i,\ti)] - \inf_{t_i \in T_i^+} \Eis[d(t_i',\ti,s)\right]\right. c_i \nonumber \\
     & \hspace{3cm}\left. -\mathbbm{1}_{T_i^-}(t_i)  \frac{1}{p}\left[ \sup_{t_i \in T_i^+} \Eis[d(t_i',\ti,s)]-\Eis[d(t_i,\ti)]\right] c_i \right]\label{eq:ineq_imp2} \\    
  &=V_{\tilde{R}}(d) \nonumber
\end{align}
\end{proof}
\normalsize

\begin{lemma}\label{lemma:imperfect_opt}
Suppose $T$ is finite. The decision rule stated in Theorem \ref{thm:imp} solves problem \eqref{relax_imp}.
\end{lemma}
\begin{proof}
Let $d^*$ denote an optimal solution to the relaxed problem \eqref{relax_imp} above, and define $\varphi_i^+\equiv \inf_{t_i'\in T_i^+}  \E_{t_{-i}}[d^*(t'_i,t_{-i})] $ and $\varphi_i^-\equiv \sup_{t_i'\in T_i^-} \E_{t_{-i}}[d^*(t'_i,t_{-i})]$. Then $d^*$ also solves the following problem:
  \begin{align*} 
    \underset{0 \leq d\leq 1 }\max \ \E_{t}  &\big[\textstyle\sum_id(t)[t_i-\frac{\tilde{c}_i(t_i)}{p}] \big]&\tag{Aux}\\
       \hspace{1cm}  &\text{s.t.  for all } i\in\I:&    \\
       \hspace{1cm}  \varphi_i^+ &\le \E_{t_{-i}} d(t) \le \frac{\varphi_i^+}{1-p} &\text{ for all } t_i\in T_i^+,\text{ and }\\
      \hspace{1cm} \frac{\varphi_i^--p}{1-p} &\le \E_{t_{-i}} d(t)\leq \varphi_i^-  &\text{ for all } t_i\in T_i^-.
  \end{align*}
The Karush-Kuhn-Tucker theorem implies that there exist Lagrange multipliers $\lambda_i(t_i)$ and $\mu_i(t_i)$ such that $d^*$ maximizes the Lagrangian:
\begin{align*}
\mathcal{L}(d,\lambda,\mu) &= \E_t\Big[\sum_i d(t)(t_i-\frac{ \tilde{c}_i(t_i)}{p})\Big] \\
&+ \sum_i \sum_{t_i \in T_i^+} \Big( \lambda_i(t_i) \big( \E_{t_{-i}}[d(t_i,t_{-i})] -\varphi_i^+ \big) +\mu_i(t_i) \big(\frac{\varphi_i^+}{1-p} -\E_{t_{-i}}[d(t_i,t_{-i})] \big)\Big) \\
&+ \sum_i\sum_{t_i \in T_i^-} \Big( \lambda_i(t_i) \big( \E_{t_{-i}}[d(t_i,t_{-i})] -\varphi_i^- \big) +\mu_i(t_i) \big(\frac{\varphi_i^--p}{1-p} -\E_{t_{-i}}[d(t_i,t_{-i})] \big)\Big) 
\end{align*}

Define $h_i(t_i) := t_i -\frac{\tilde{c}_i(t_i)}{p} + \frac{\lambda_i(t_{i})+\mu_i(t_i)}{f_i(t_i)}$ and let
\begin{eqnarray*}
  \alpha_i^+ =& \inf_{t_i \in T_i^+}\{ t_i | \E_{t_{-i}}[d^*(t_i,t_{-i})]  > \varphi_i^+ \}  \\
  \alpha_i^- =& \sup_{t_i \in T_i^-}\{ t_i | \E_{t_{-i}}[d^*(t_i,t_{-i})] < \varphi_i^- \} \\
  \beta_i^+ =& \sup_{t_i \in T_i^+}\{t_i | \E_{t_{-i}}[d^*(t_i,t_{-i})]  < \frac{\varphi_i^+}{1-p} \} \\
  \beta_i^- =& \inf_{t_i \in T_i^-}\{ t_i | \E_{t_{-i}}[d^*(t_i,t_{-i})] > \frac{\varphi_i^--p}{1-p} \}.
  \end{eqnarray*}

Define $A_i^+ =\{t_i \in T_i^+| t_i<\alpha_i^+\}$, $A_i^-=\{t_i \in T_i^-| t_i>\alpha_i^-\}$, $B_i^+ = \{t_i \in T_i^+| t_i>\beta_i^+\}$, $B_i^- = \{t_i \in T_i^-| t_i<\beta_i^-\}$, and 
  \begin{align*}
    \bar{h}_i(t_i) := \begin{cases}
      \frac{1}{\mu_i(A_i^+)} \sum_{t_i\in A_i^+} f_i(t_i) h_i(t_i) \hspace{.3cm} &\text{ if } t_i \in A_i^+  \\
      \frac{1}{\mu_i(B_i^+)} \sum_{t_i\in B_i^+} f_i(t_i) h_i(t_i) \hspace{.3cm} &\text{ if } t_i \in B_i^+  \\
     \frac{1}{\mu_i(A_i^-)} \sum_{t_i\in A_i^-} f_i(t_i) h_i(t_i) \hspace{.3cm} &\text{ if } t_i \in A_i^-  \\
      \frac{1}{\mu_i(B_i^-)} \sum_{t_i\in B_i^-} f_i(t_i) h_i(t_i) \hspace{.3cm} &\text{ if } t_i \in B_i^-  \\
      t_i-\tilde{c}_i(t_i) &\text{ otherwise}.
    \end{cases}
  \end{align*}
The same arguments as in the proof of Lemma \ref{prop:opt} imply that $d^*$ maximizes $\sum_i \sum_t f(t)d(t)\bar{h}_i(t_i)$. 
\end{proof}

\begin{lemma}\label{lemma:imp_feas}
Suppose $T$ is infinite. The decision rule stated in Theorem \ref{thm:imp} solves problem \eqref{relax_imp}.
\end{lemma}
\begin{proof}
The proof is analogous to the proof of Lemma \ref{prop:opt_inf}  and hence omitted.
\end{proof}

\begin{proof}[Proof of Theorem \ref{thm:imp}]
Denote by $d^*$ the solution to problem \ref{relax_imp}. For each $i$, define $q_i:T_i\rightarrow [0,1]$ as the solution to 
\begin{align*}
\Ei[d^*(t_i,\ti)\ [1-p\cdot q_i(t_i)]] &= \inf_{t_i'\in T_i^+} \Ei[d^*(t_i',\ti)] &\text{ , for } t_i\in T_i^+ \text{ and}\\
\E_{t_{-i}}[d^*(t_i,t_{-i}) [1-p\cdot q_i(t_i)] &= \sup_{t_i'\in T_i^-} \E_{t_{-i}}[d^*(t_i',t_{-i})] - p \cdot q_i(t_i)] &\text{ , for } t_i\in T_i^-.
\end{align*}
We will now show that a solution $q_i$ exists.
For $t_i\in T_i^+$, setting $q_i(t_i)=0$ yields
\[ \Ei[d^*(t_i,\ti)[1-pq_i(t_i)]] = \Ei[d^*(t_i,\ti)] \ge \inf_{t_i'\in T_i^+} \Ei[d^*(t_i',\ti)] \]
and setting $q_i(t_i)=1$ yields
\[ \Ei[d^*(t_i,\ti)[1-pq_i(t_i)]] = \Ei[d^*(t_i,\ti)[1-p]] \le \inf_{t_i'\in T_i^+} \Ei[d^*(t_i',\ti)], \]
where the inequality follows from \eqref{eq:bic_imp1}. The intermediate-value theorem hence implies the existence of a solution $q_i$. Similar arguments apply for $t_i\in T_i^-$.

Define
\begin{align*}
a^*_i(t):= \begin{cases}
q_i(t_i) \hspace{.5cm} &\text{ if $t_i\in T_i^+$ and $d^*(t)=1$}\\
q_i(t_i) &\text{ if $t_i\in T_i^-$ and $d^*(t)=0$}\\
0 &\text{ else.}
\end{cases}
\end{align*}
For each $i$ and for all $t_i \in T_i^+$,
\[ \inf_{t_i'\in T_i^+} \E_{t_{-i},s}[d^*(t_i',t_{-i},s)] = \E_{t_{-i},s}[d^*(t_i,t_{-i},s) [1-p\cdot a^*_i(t_i,t_{-i},s)]], \]
and for all $t_i \in T_i^-$,
\[ \sup_{t_i'\in T_i^-} \E_{t_{-i},s}[d^*(t_i',t_{-i},s)] = \E_{t_{-i},s}[d^*(t_i,t_{-i},s) [1-p\cdot a^*_i(t_i,t_{-i},s)] + p \cdot a^*_i(t_i,t_{-i},s)]. \]
Hence, $(d^*,a^*)$ is Bayesian incentive compatible by Lemma \ref{lemma_bic_impp} and inequality \eqref{eq:ineq_imp2} holds as an equality.
By construction, $t_i\in T_i^+$ implies $d(t)a^*_i(t)=a^*_i(t)$ and $t_i\in T_i^-$ implies $[1-d(t)]a^*_i(t)=a^*_i(t)$. Therefore, inequality \eqref{eq:ineq_imp1} also holds as an equality and we conclude $V_{\tilde{P}}(d^*,a^*)=V_{\tilde{R}}(d^*)$. Hence,  $(d^*,a^*)$ is optimal.
\end{proof}

\end{appendix}

\begin{small}
	\bibliographystyle{elsarticle-harv}
	\bibliography{references}
\end{small}

\end{document}